\documentclass[11pt]{amsart}
\usepackage{marvosym}
\usepackage{slashed}
\usepackage{hyperref}
\usepackage{amsthm}
\usepackage{epsfig,subfigure,latexsym,dcolumn,amsmath,longtable,setspace}
\usepackage{amsmath}
\usepackage{amsfonts}
\usepackage{amssymb}
\usepackage{latexsym}
\usepackage{mathrsfs}
\usepackage{amsmath,amscd}
\usepackage{amsmath, amscd, amssymb}
\usepackage{graphicx, psfrag}
\usepackage[frame,cmtip,arrow,matrix,line,graph,curve,color]{xy}
\usepackage{graphpap, color}
\usepackage[mathscr]{eucal}
\usepackage{mathrsfs}
\usepackage{tikz}
\usepackage{xcomment}
\usepackage{verbatim}

\usetikzlibrary{matrix,arrows}

\numberwithin{equation}{section}

\newtheorem{prop}{Proposition}[section]
\newtheorem{theo}[prop]{Theorem}
\newtheorem{lemm}[prop]{Lemma}

\newtheorem{defi}[prop]{Definition}

\title[Quantitative Mode Stability for the Wave Equation]{Quantitative Mode Stability for the Wave Equation on the Kerr Spacetime}
\author{Yakov Shlapentokh-Rothman}

\thanks{This work was partially supported by NSF grant DMS-0943787.}
\address{Department of Mathematics\\  Massachusetts Institute of Technology\\ Cambridge, MA 02139\\ USA}

\email{yakovsr@math.mit.edu}
\begin{document}
\begin{abstract}
We give a quantitative refinement and simple proofs of mode stability type statements for the wave equation on Kerr backgrounds in the full sub-extremal range $(|a| < M)$. As an application, we are able to quantitatively control the energy flux along the horizon and null infinity and establish integrated local energy decay for solutions to the wave equation in any bounded-frequency regime.
\end{abstract}
\maketitle
\tableofcontents
\section{Introduction}
One of the most central problems in mathematical General Relativity is the non-linear stability of the $2$-parameter family of Kerr spacetimes $(\mathcal{M},g_{a,M})$, indexed by mass $M$ and specific angular momentum $a$. Though the full non-linear problem (the stability of $(\mathcal{M},g_{a,M})$ as a family of solutions to the Einstein vacuum equations Ric$(g) = 0$) appears intractable at the moment, much work has been done in the linear setting. In particular, experience teaches us that resolving the non-linear problem will require a robust understanding of decay for solutions of the wave equation $\Box_g\psi = 0$ on the fixed Kerr spacetime $(\mathcal{M},g)$. Let us direct the reader to the lecture notes \cite{n6} for a general introduction to linear waves on black hole backgrounds.

Surprisingly, even the most basic boundedness and decay statements for the wave equation on Kerr remained unanswered until quite recently. Boundedness and decay results for solutions to the wave equation on the $1$-parameter Schwarzschild subfamily $(a = 0)$ were obtained in \cite{n18},\cite{n19}, and \cite{n20}. The first global result for general solutions to the Cauchy problem on a rotating black hole $(a \neq 0)$ was obtained in \cite{n3} where Dafermos and Rodnianski established uniform boundedness in the case $|a|\ll M$. Following this, decay results, again in the case $|a| \ll M$, were obtained by various authors, e.g. \cite{n6}, \cite{n4}, \cite{n5}, \cite{n9}, \cite{n10}, \cite{n11}, and \cite{n12}. For the full sub-extremal range of Kerr black holes $(|a| < M)$, the coupling between ``superradiance'' and trapping presented serious conceptual difficulties; nevertheless, in \cite{n2} Dafermos and Rodnianski succeeded in establishing boundedness and decay for the wave equation on a general sub-extremal Kerr background. Their proof required an additional estimate\footnote{See their discussion in section 11.7 of \cite{n2}.} for the ``bounded superradiant frequencies.'' This paper provides the needed result.

Interestingly, the problem of the superradiant frequencies will lead us back to the classical mode analysis of the physics literature, see \cite{n13} and \cite{n1}, albeit from a quite different perspective. Mode solutions to the wave equation will be reviewed in section \ref{reviewModes}; for now, we simply recall that a solution $\psi$ to the wave equation $\Box_g\psi = 0$ is called a mode solution if
\[\psi(t,r,\theta,\phi) = e^{-i\omega t}e^{im\phi}S(\theta)R(r)\text{ with }\omega \in \mathbb{C}\text{ and }m \in \mathbb{Z},\]
where $(t,r,\theta,\phi)$ are Boyer-Lindquist coordinates (defined in section \ref{notation}) and $S$ and $R$ must satisfy appropriate ordinary differential equations and boundary conditions (given in section \ref{reviewModes}) so that, among other things, $\psi$ has finite energy along suitable spacelike hypersurfaces.\footnote{When $\text{Im}\left(\omega\right) > 0$ one should take asymptotically flat hypersurfaces connecting the future event horizon and spacelike infinity. When $\text{Im}\left(\omega\right) \leq 0$ one should instead consider hyperboloidal hypersurfaces connecting the future event horizon and future null infinity. See the discussion in appendix \ref{finiteModeEnergy}.} Ruling out the exponentially growing mode solutions corresponding to Im$(\omega) > 0$ is the content of ``mode stability.'' This was established by Whiting in the ground-breaking \cite{n1}. We will extend Whiting's techniques and establish a \emph{quantitative} understanding of the lack of mode solutions with \emph{real} $\omega$.\footnote{See also \cite{n15} and \cite{n16} which concern solutions to the Cauchy problem of the form $e^{im\phi}\psi_0(t,r,\theta)$ and discuss mode solutions with real $\omega$.} As a byproduct of our methods, we will also be able to simplify the proof of Whiting's original mode stability result. Next, we will show that this ``quantitative mode stability on the real axis'' can be upgraded to ``integrated local energy decay,'' with an explicit constant, for solutions to the wave equation in any ``bounded-frequency regime.''\footnote{The phrase ``bounded-frequency regime'' will be precisely defined in section \ref{theResults}; but, lest the reader be mislead, we take the opportunity to emphasize that the integrated energy decay statement proven here will assume \emph{a priori} that the solution and its coordinate derivatives are square integrable to the future in $t$.} Along the way, we will produce the necessary estimate for section 11.7 of \cite{n2}.
\subsection{The Spacetime}\label{notation}
Fix a pair of parameters $(a,M)$ with $|a| < M$, and define
\[r_+ := M + \sqrt{M^2 - a^2}.\]
Define the underlying manifold $\mathcal{M}$ to be covered by a global\footnote{``Global'' is be understood with respect to the usual degeneracy of polar coordinates.} ``Boyer-Lindquist'' coordinate chart
\[(t,r,\theta,\phi) \in \mathbb{R} \times (r_+,\infty)\times \mathbb{S}^2.\]
The Kerr metric then takes the form
\begin{equation}\label{metric}
g_{a,M} = -\left(1-\frac{2Mr}{\rho^2}\right)dt^2 - \frac{4Mar\sin^2\theta}{\rho^2}dtd\phi + \frac{\rho^2}{\Delta}dr^2
\end{equation}
\[+ \rho^2 d\theta^2 + \sin^2\theta\frac{\Pi}{\rho^2}d\phi^2,\]
\[r_{\pm} := M \pm \sqrt{M^2-a^2},\]
\[\Delta := r^2 - 2Mr + a^2 = (r-r_+)(r-r_-),\]
\[\rho^2 := r^2 + a^2\cos^2\theta,\]
\[\Pi := (r^2+a^2)^2 - a^2\sin^2\theta\Delta.\]
It is convenient to define an $r^*(r) :(r_+,\infty) \to (-\infty,\infty)$ coordinate up to a constant by
\[\frac{dr^*}{dr} := \frac{r^2+a^2}{\Delta}.\]
We will often drop the parameters and refer to $g_{a,M}$ as $g$.

It turns out that the manifold $\mathcal{M}$ can be extended to a manifold $\tilde{\mathcal{M}}$ such that $\partial \mathcal{M}$ is a null hypersurface called the ``horizon.'' Since Boyer-Lindquist coordinates would break down at the horizon, one needs a new coordinate system. The standard choice is ``Kerr-star" coordinates $(t^*,r,\phi^*,\theta)$:
\[\frac{d\overline{t}}{dr} := \frac{r^2+a^2}{\Delta},\]
\[\frac{d\overline{\phi}}{dr} := \frac{a}{\Delta},\]
\[t^*(t,r) := t + \overline{t}(r),\]
\[\phi^*(\phi,r) := \phi + \overline{\phi}(r).\]
In these coordinates the metric becomes
\[g = -\left(1-\frac{2Mr}{\rho^2}\right)(dt^*)^2 - \frac{4Mar\sin^2\theta}{\rho^2}dt^*d\phi^* + 2dt^*dr + \]
\[\rho^2 d\theta^2 + \sin^2\theta\frac{\Pi}{\rho^2}(d\phi^*)^2 -2a\sin^2\theta drd\phi^*.\]
Note that we can now extend the metric to the manifold $\tilde{\mathcal{M}} := (t^*,r,\theta,\phi^*) \in \mathbb{R}\times (0,\infty)\times \mathbb{S}^2$. The (future) event horizon $\mathcal{H}^+$ is defined to be the null hypersurface $\{r = r_+\}$.

\subsection{Separating the Wave Equation: Mode Solutions}\label{reviewModes}
When $a = 0$, in addition to possessing the Killing vector field $\partial_t$, the metric (\ref{metric}) is spherically symmetric. Thus, it is immediately clear that the wave equation $\Box_{g_{0,M}}\psi = 0$ is separable. When $a \neq 0$ the only Killing vector fields are $\partial_t$ and $\partial_{\phi}$. Nevertheless, as first discovered by Carter \cite{n22}, the wave equation $\Box_g\psi = 0$ remains separable (in an appropriate coordinate system). Indeed, letting $(\omega,m) \in \mathbb{C}\setminus\{0\}\times\mathbb{Z}$, we have
\[\frac{e^{i\omega t}e^{-im\phi}}{\rho^2}\Box_g\left(e^{-i\omega t}e^{im\phi}\psi_0(r,\theta)\right) = \]
\begin{equation}\label{compute}
\partial_r\left(\Delta\partial_r\right)\psi_0 + \left(\frac{(r^2+a^2)\omega^2- 4Mamr\omega + a^2m^2}{\Delta} - a^2\omega^2\right)\psi_0 +
\end{equation}
\[\frac{1}{\sin\theta}\partial_{\theta}\left(\sin\theta\partial_{\theta}\right)\psi_0 - \left(\frac{m^2}{\sin^2\theta}-a^2\omega^2\cos^2\theta\right)\psi_0.\]
In fact, the separability of the wave equation follows from the presence on Kerr of a \emph{Killing tensor} \cite{n23}.

We call
\begin{equation}\label{sml}
    \frac{1}{\sin\theta}\frac{d}{d\theta}\left(\sin\theta\frac{dS}{d\theta}\right) - \left(\frac{m^2}{\sin^2\theta} - a^2\omega^2\cos^2\theta\right)S + \lambda S = 0
\end{equation}
the ``angular ODE.'' One can show that when $\omega \in \mathbb{R}$, then (\ref{sml}) along with the boundary condition
\begin{equation}\label{boundSml}
    e^{im\phi}S(\theta) \text{ extends smoothly to }\mathbb{S}^2
\end{equation}
defines a Sturm-Liouville problem with a corresponding collection of eigenfunctions $\{S_{\omega ml}\}_{l=|m|}^{\infty}$ and real eigenvalues $\{\lambda_{\omega ml}\}_{l=|m|}^{\infty}$. These $\{S_{\omega ml}\}$ are an orthonormal basis of $L^2(\sin\theta d\theta)$ and are called ``oblate spheroidal harmonics.'' When $a = 0$ these are simply spherical harmonics, and we label them in the standard way so that $\lambda_{\omega ml} = l(l+1)$. For $a \neq 0$, the labeling is uniquely determined by enforcing continuity in $a$. Lastly, we note that for $\omega$ with sufficiently small imaginary part, one may define the $S_{\omega ml}$ and $\lambda_{\omega ml}$ via perturbation theory \cite{n14}.

Now we are ready for the main definition of the section.
\begin{defi}\label{mode}Let $(\mathcal{M},g)$ be a sub-extremal Kerr spacetime with parameters $(a,M)$. A smooth solution $\psi$ to the wave equation
\begin{equation}\label{waveEqn}
\Box_g\psi = 0
\end{equation}
is called a ``mode solution'' if there exist ``parameters'' $(\omega,m,l) \in \mathbb{C}\setminus\{0\}\times \mathbb{Z}\times\mathbb{Z}_{\geq m}$ such that
\begin{equation}\label{modeForm}
\psi(t,r,\theta,\phi) = e^{-i\omega t}e^{im\phi}S_{\omega ml}(\theta)R(r,\omega,m,l),
\end{equation}
where
\begin{enumerate}
\item $S_{\omega ml}$ satisfies the boundary condition (\ref{boundSml}) and is an eigenfunction with eigenvalue $\lambda_{\omega ml}$ for the angular ODE (\ref{sml}).
\item $R$ is a solution to
    \begin{equation}\label{firstEqn}
    \partial_r\left(\Delta\partial_r\right)R + \left(\frac{(r^2+a^2)\omega^2- 4Mamr\omega + a^2m^2}{\Delta} - \lambda_{\omega ml} - a^2\omega^2\right)R = 0
    \end{equation}
\item \begin{equation}\label{horizonBound}
        R \sim (r-r_+)^{\frac{i(am-2Mr_+\omega)}{r_+-r_-}}\text{ at }r = r_+.\footnote{This notation means that $R(r)(r-r_+)^{\frac{-i(am-2Mr_+\omega)}{r_+-r_-}}$ is smooth at $r = r_+$.}
    \end{equation}
\item \begin{equation} \label{infinityBound}
        R \sim \frac{e^{i\omega r^*}}{r}\text{ at }r = \infty.\footnote{This notation means that there exists constants $\left\{C_i\right\}_{i=0}^{\infty}$ such that for every $N \geq 1$, $R(r^*) = \frac{e^{i\omega r^*}}{r}\sum_{i=0}^N\frac{C_i}{r^i} + O\left(\left(r\right)^{-N-2}\right)\text{ for large }r$.}
    \end{equation}
\end{enumerate}
\end{defi}
We will often suppress some of the arguments of $S_{\omega ml}$ and $R$ and refer to them as $S_{\omega ml}(\theta)$ and $R(r)$.

Instead of considering $R(r)$, it is often more convenient to work with the function
\[u(r^*) := (r^2+a^2)^{1/2}R(r).\]
Then, letting primes denote $r^*$-derivatives, equation (\ref{firstEqn}) is equivalent to
\begin{equation}\label{eqn0}
u'' + \left(\omega^2-V\right)u = 0,
\end{equation}
\begin{align*}
V :=\ &\frac{4Mram\omega - a^2m^2 + \Delta(\lambda_{\omega ml} + a^2\omega^2)}{(r^2+a^2)^2}\\
&+ \frac{\Delta}{(r^2+a^2)^4}\left(a^2\Delta + 2Mr(r^2-a^2)\right).
\end{align*}
In appendix \ref{odeFacts} we have collected various facts about the relevant class of ODEs that will be used throughout the paper. The boundary conditions given for $R$ and $S_{\omega ml}$ ((\ref{horizonBound}), (\ref{infinityBound}), and (\ref{boundSml})) are uniquely determined by requiring that $\psi$, given by (\ref{modeForm}), extends smoothly to the horizon, has finite energy along asymptotically flat hypersurfaces when $\text{Im}\left(\omega\right) > 0$, and has finite energy along hyperboloidal hypersurfaces when $\text{Im}\left(\omega\right) \leq 0$ (see the discussion in appendix \ref{finiteModeEnergy}). Furthermore, in section \ref{theBoundary} we will see these boundary conditions directly arise during the proof of integrated local energy decay. Though they will only concern us tangentially here, it is worth mentioning that there is a large literature devoted to locating mode solutions with $\text{Im}\left(\omega\right) < 0$ (see the review \cite{n13}). These are called \emph{quasi-normal modes} and are expected to provide to great deal of dynamical information about the decay of scalar fields. For a sample of the mathematical study of quasi-normal modes and corresponding applications (to black hole spacetimes), we recommend \cite{n24}, \cite{n25}, \cite{n26}, \cite{n27}, \cite{n28}, \cite{n29}, \cite{n30}, \cite{n31}, \cite{n32}, and the references therein.
\subsection{Mode Stability Type Statements}
Ruling out the exponentially growing mode solutions corresponding to $\text{Im}\left(\omega\right) > 0$ is the content of ``mode stability (in the upper half plane).'' This was established by Whiting in 1989 \cite{n1}. However, this turns out not to be the full story. Indeed, the existence of mode solutions with $\omega \in \mathbb{R}\setminus\{0\}$ is a serious obstruction to ``integrated local energy decay'' for the wave equation. We will call the ruling out of these mode solutions ``mode stability on the real axis.'' This was first explored numerically in \cite{n21}. In addition, \cite{n21} presented a heuristic argument (rigorously established in \cite{n17}) indicating that mode stability on the real axis would imply mode stability in the upper half plane. In section \ref{wronskianReduct} we will show how one can upgrade mode stability on the real axis to \emph{integrated local energy decay} for the wave equation in any ``bounded-frequency regime.'' In order for the constant in this estimate to be explicit, however, we will be interested in a quantitative version of mode stability of the real axis.

We turn now to an explanation of ``quantitative mode stability.'' Observe that if a solution to the angular ODE exists, an asymptotic analysis of (\ref{eqn0}) (see appendix \ref{odeFacts}) allows one to make the following definitions:
\begin{defi}\label{theHor}Let the parameters $|a| < M$ be fixed. Then define $u_{\text{hor}}(r^*,\omega,m,l)$ to be the unique function satisfying
    \begin{enumerate}
        \item $u_{\text{hor}}'' + \left(\omega^2-V\right)u_{\text{hor}} = 0$.
        \item $u_{\text{hor}} \sim (r-r_+)^{\frac{i(am-2Mr_+\omega)}{r_+-r_-}}\text{ near }r^* = -\infty$.
        \item $\left|\left((r(r^*)-r_+)^{\frac{-i(am-2Mr_+\omega)}{r_+-r_-}}u_{\text{hor}}\right)\left(-\infty\right)\right|^2 = 1$.
    \end{enumerate}
\end{defi}
\begin{defi}\label{theOut}Let the parameters $|a| < M$ be fixed. Then define $u_{\text{out}}(r^*,\omega,m,l)$ to be the unique function satisfying
    \begin{enumerate}
        \item $u_{\text{out}}'' + \left(\omega^2-V\right)u_{\text{out}} = 0$.
        \item $u_{\text{out}} \sim e^{i\omega r^*}\text{ near }r^* = \infty$.
        \item $\left|\left(e^{-i\omega r^*}u_{\text{out}}\right)\left(\infty\right)\right|^2 = 1$.
    \end{enumerate}
\end{defi}
See appendix \ref{odeFacts} for the explicit definition of ``$\sim$''. When there is no risk of confusion, we shall drop some or all of $u_{\text{hor}}$'s and $u_{\text{out}}$'s arguments. Next, recall that the \emph{Wronskian}
\[u_{\text{out}}'(r^*)u_{\text{hor}}(r^*)-u_{\text{hor}}'(r^*)u_{\text{out}}(r^*)\]
is independent of $r^*$. Hence, we can define
\begin{equation}\label{theWronk}
W(\omega,m,l) := u_{\text{out}}'(r^*)u_{\text{hor}}(r^*)-u_{\text{hor}}'(r^*)u_{\text{out}}(r^*).
\end{equation}
This will vanish if and only if $u_{\text{out}}$ and $u_{\text{hor}}$ are linearly dependent, i.e.~there exists a non-trivial solution to (\ref{eqn0}) $\Leftrightarrow W = 0 \Leftrightarrow \left|W^{-1}\right| = \infty$. ``Quantitative mode stability'' consists of producing an upper bound for $\left|W^{-1}\right|$ with an explicit dependence on $a$, $M$, $\omega$, $m$, and $l$.
\subsection{Statement of Results}\label{theResults}
Fix a Kerr spacetime $(\mathcal{M},g)$ with parameters $(a,M)$ satisfying $|a| < M$, and recall the definition of mode solutions (definition \ref{mode}) and the Wronskian (\ref{theWronk}) given in the previous section.

Our main result about mode solutions is
\begin{theo}(Quantitative Mode Stability on the Real Axis)\label{qMSRA} Let
\[\mathscr{A} \subset \left\{(\omega,m,l) \in \mathbb{R}\times\mathbb{Z}\times\mathbb{Z}_{\geq |m|}\right\}\]
be a set of frequency parameters with
\[C_{\mathscr{A}} := \sup_{(\omega,m,l)\in \mathscr{A}}\left(\left|\omega\right| + \left|\omega\right|^{-1} + \left|m\right| + |l|\right) < \infty.\]
Then
\[\sup_{(\omega,m,l)\in\mathscr{A}}\left|W^{-1}\right| \leq G(C_{\mathscr{A}},a,M)\]
where the function $G$ can, in principle, be given explicitly.

\end{theo}
Along the way we will give simple\footnote{Using Whiting's integral transformations \cite{n1} but avoiding differential transformations or a physical space argument with a new metric.} proofs of
\begin{theo}(Mode Stability)(Whiting \cite{n1})\label{modeStability} There exist no non-trivial mode solutions corresponding to $\text{Im}\left(\omega\right) > 0$.
\end{theo}
\begin{theo}(Mode Stability on the Real Axis)\label{realModeStability} There exist no non-trivial mode solutions corresponding to $\omega \in \mathbb{R}\setminus\{0\}$.
\end{theo}

Before discussing our main application, we need a few definitions.
\begin{defi}\label{admissible}We will say that a $C^{\infty}\left(\mathcal{M}\right)$ function $\psi\left(t,r,\theta,\phi\right)$ is admissible if
    \begin{enumerate}
        \item For each multi-index $\alpha$ with $\left|\alpha\right| \geq 1$, and sufficiently large $r_0$, we have
            \[\int_{r > r_0}\int_{\mathbb{S}^2}\left|\partial^{\alpha}\psi\right|^2\Big|_{t = 0}r^2\sin\theta\ dr\ d\theta\ d\phi < \infty.\]
        \item For each multi-index $\alpha$ with $\left|\alpha\right| \geq 0$, and Boyer-Lindquist $(r,\theta,\phi) \in \left(r_+,\infty\right)\times \mathbb{S}^2$, we have
            \[\int_0^{\infty}\left|\partial^{\alpha}\psi\right|^2dt < \infty.\]
        \item For every compact $K \in \left(r_+,\infty\right)\times\mathbb{S}^2$ and multi-index $\alpha$ with $\left|\alpha\right| \geq 0$, we have
            \[\int_0^{\infty}\int_K\left|\partial^{\alpha}\psi\right|^2\sin\theta\ dt\ dr\ d\theta\ d\phi < \infty.\]
    \end{enumerate}
    All of these derivatives are Boyer-Lindquist derivatives.
\end{defi}
\begin{defi}\label{boundedFreq}
Let $h$ be an admissible function on Kerr. Let $\mathscr{B} \subset \mathbb{R}$ and $\mathscr{C} \subset \{(m,l) \in \mathbb{Z}\times \mathbb{Z} : l \geq |m|\}$ be such that
\[C_{\mathscr{B}} := \sup_{\omega \in \mathscr{B}}\left(\left|\omega\right| + \left|\omega\right|^{-1}\right) < \infty,\]
\[C_{\mathscr{C}} := \sup_{m,l \in \mathscr{C}}\left(\left|m\right| + \left|l\right|\right) < \infty.\]
Then we define
\[P_{\mathscr{B},\mathscr{C}}h := \]
\[\int_{\mathscr{B}}\sum_{(m,l) \in \mathscr{C}}\left(\int_0^{\pi}\int_0^{2\pi}\int_{-\infty}^{\infty}h e^{i\omega \tau}e^{-im\varphi}S_{\omega ml}\sin\vartheta\ d\tau\ d\varphi\ d\vartheta\right)S_{\omega ml}e^{im\phi}e^{-i\omega t}\ d\omega.\]
\end{defi}
Next, let $\Sigma_0$ be a spacelike hyperboloidal\footnote{The hyperboloidal condition will be satisfied if for sufficiently large $r$, the hypersurface $\Sigma_0$ is spacelike, given by the zero set of $t - f(r^*)$, and $f$ satisfies
\[\left(f'\right)^2 - 1 = -\frac{C}{r^2} + O\left(r^{-3}\right)\text{ as }r\to\infty.\]
See the discussion in appendix \ref{finiteModeEnergy}.} hypersurface connecting the future event horizon $\mathcal{H}^+$ and future null infinity. The relevant Penrose diagram is given by
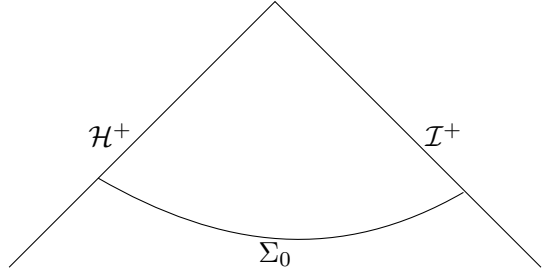
\begin{center}
    \begin{tikzpicture}
        [interior/.style={circle,draw=black,fill=black, inner sep=0pt,minimum size = 2.5mm},
        boundary/.style={circle,draw=black,fill=black!60, inner sep=0pt,minimum size = 2.5mm},
        exterior/.style={circle,draw=black,fill=white, inner sep=0pt,minimum size = 2.5mm},
        highlight/.style = {circle,draw= red, fill = red, inner sep=0pt, minimum size = 2.5mm}]
        \draw (0,0) -- (3.535,3.535) -- (7.071,0);
        \node [left] at (1.768,1.768) {$\mathcal{H}^+$};
        \node [right] at (5.4,1.768) {$\mathcal{I}^+$};
        \draw (1.184,1.186) to [out = -30,in = 210] (6.04,.996);
        \node [below] at (3.536,.47) {$\Sigma_0$};
    \end{tikzpicture}
\end{center}

Let $\Sigma_1$ be the image of $\Sigma_0$ under the time $1$ map of the flow generated by $\partial_t$. Define a cutoff $\chi$ which is $0$ in the past of $\Sigma_0$ and identically $1$ in the future of $\Sigma_1$.

Our application of Theorem \ref{qMSRA} will be
\begin{theo}(Boundedness of the Microlocal Energy Flux and Integrated Local Energy Decay in the Bounded-Frequency Regime)\label{boundMicroLocalEnergy} Let $\psi$ be an admissible\footnote{This condition could be relaxed considerably; however, our main goal is to simply give a flavor of the sort of results which follow from Theorem \ref{qMSRA}.} function on Kerr that is also a solution to the wave equation $\Box_g\psi = 0$. Set
\[\psi_{\text{\Rightscissors}} := \chi\psi.\]

Let $\mathscr{B} \subset \mathbb{R}$ and $\mathscr{C} \subset \{(m,l) \in \mathbb{Z}\times \mathbb{Z} : l \geq |m|\}$ be such that
\[C_{\mathscr{B}} := \sup_{\omega \in \mathscr{B}}\left(\left|\omega\right| + \left|\omega\right|^{-1}\right) < \infty\]
\[C_{\mathscr{C}} := \sup_{m,l \in \mathscr{C}}\left(\left|m\right| + \left|l\right|\right) < \infty.\]

Then, for every $r_+ < r_0 < r_1 < \infty$,
\begin{equation}\label{estHorizon}
\int_{\mathcal{H}^+}\left|P_{\mathscr{B},\mathscr{C}}\psi_{\text{\Rightscissors}}\right|^2
\end{equation}
\[+ \int_{\mathcal{I}^+}\left|\partial P_{\mathscr{B},\mathscr{C}}\psi_{\text{\Rightscissors}}\right|^2 + \int_{\mathbb{R}\times (r_0,r_1)\times\mathbb{S}^2}\left|\partial P_{\mathscr{B},\mathscr{C}}\psi_{\text{\Rightscissors}}\right|^2 \leq\]
\[B\left(r_0,r_1,C_{\mathscr{B}},C_{\mathscr{C}},a,M\right) \int_{\Sigma_0}\left|\partial\psi\right|^2\]
where $\left|\partial\psi\right|^2$ denotes a term proportional to a non-degenerate energy flux of a globally timelike vector field (see appendix \ref{theFirstt}). In particular, this energy will degenerate as $r \to \infty$ due to the hyperboloidal nature of $\Sigma_0$:
\begin{align*}
\left|\partial\psi\right|^2_{\Sigma_0} \approx &(\left(\partial_t+\partial_{r^*}\right)\psi)^2 + r^{-2}(\left(\partial_t-\partial_{r^*}\right)\psi)^2 \\
&+ r^{-2}\left(\sin^{-2}\theta(\partial_{\phi}\psi)^2 + (\partial_{\theta}\psi)^2\right)\text{ as }r\to\infty
\end{align*}
The energy at future null infinity is explicitly given by
\[\left|\partial P_{\mathscr{B},\mathscr{C}}\psi_{\text{\Rightscissors}}\right|^2_{\mathcal{I}^+} \approx \lim_{r\to\infty}r^2\left(\left|\partial_t P_{\mathscr{B},\mathscr{C}}\psi_{\text{\Rightscissors}}\right|^2 + \left|\partial_r P_{\mathscr{B},\mathscr{C}}\psi_{\text{\Rightscissors}}\right|^2\right).\]
Note that the spacetime volume form satisfies
\[\text{dVol}_{(t,r,\theta,\phi)} \approx r^2\sin\theta\, dt\, dr\, d\theta\, d\phi.\]
The function $B\left(r_0,r_1,C_{\mathscr{B}},C_{\mathscr{C}},a,M\right)$ can, in principle, be given explicitly.
\end{theo}

Of course, since we are in a bounded-frequency regime, the zeroth order estimate along the horizon (\ref{estHorizon}) controls the microlocal energy flux along the horizon:
\[\int_{\mathscr{B}}\sum_{(m,l)\in\mathscr{C}}\omega\left(am-2Mr_+\omega\right)\left|u(-\infty)\right|^2d\omega.\]
Here
\[u(r^*) := (r^2+a^2)^{1/2}R(r)\]
where $R(r,\omega,m,l)$ is the projection of the Fourier transform in $t$ of $\psi_{\text{\Rightscissors}}$ onto the oblate spheroidal harmonics $S_{\omega ml}$, i.e.
\[R(r) := \int_0^{\pi}\int_0^{2\pi}\int_{-\infty}^{\infty}\psi_{\text{\Rightscissors}}e^{i\omega t}e^{-im\phi}S_{\omega ml}\sin\theta\ dt\ d\phi\ d\theta.\]
The estimate for this term is utilized in Dafermos' and Rodnianski's proof of integrated local energy decay for the wave equation \cite{n2}. For this application, it is very important that the right hand side is at the level of energy.

Before diving into the proofs of our results, we will review the case of mode solutions on Schwarzschild $(a = 0)$ and what is already known about mode solutions on Kerr.
\subsection{Modes on Schwarzschild}\label{modeSchw}
It is instructive to observe that the counterpart to mode stability in the Riemannian setting\footnote{This is the case of a product metric $(\mathbb{R}\times N, -dt^2 + g_N)$ with $(N,g_N)$ complete and Riemannian.} is the ``automatic'' fact that the Laplace-Beltrami operator has no spectrum in the upper half plane. A better way to see the triviality of Riemannian mode stability is to note that the existence of a uniformly timelike vector field $\partial_t$ immediately implies the uniform boundedness of a non-degenerate energy \cite{n8}.

Recall that the Schwarzschild spacetime is the Kerr spacetime with vanishing angular momentum $(a = 0)$. This is not a product metric; nevertheless, $\partial_t$ is a timelike Killing vector field for all $r > r_+$, the associated conserved energy is coercive, and mode stability is immediately established in a similar fashion to the previous paragraph.\footnote{Of course, $\partial_t$ becomes null on the horizon, and thus the conserved energy degenerates as $r \to r_+$. However, a moment's thought shows that this does not affect the argument.}

Mode stability on the real axis for Schwarzschild is more subtle since real mode solutions have infinite energy along asymptotically flat hypersurfaces. However, this does not preclude physical space methods; one simply observes
\begin{enumerate}
    \item The boundary conditions at infinity and the horizon imply that real mode solutions have finite energy along the hypersurface $\Sigma_0$ (see appendix \ref{finiteModeEnergy}).
    \item A straightforward computation shows that the energy flux for such real modes along the portion of null infinity in the future of $\Sigma_0$ must be infinite.
    \item The energy identity associated to $\partial_t$ implies that the energy flux along the portion of null infinity in the future of $\Sigma_0$ must be less than or equal to the energy flux along $\Sigma_0$.
\end{enumerate}
This is a clear contradiction to the existence of real modes.

For later purposes it will be convenient to revisit these arguments from a ``microlocal'' point of view. In phase space, the analogue of the energy flux is the microlocal energy current:
\[Q_T(r^*) := \text{Im}\left(u'\overline{\omega u}\right).\]
Let us show how the microlocal energy can be used to give a short proof of mode stability. Suppose we have a mode solution with corresponding $u(r^*)$ and $\omega = \omega_R + i\omega_I$ for some $\omega_I > 0$. First, we observe that the boundary conditions (\ref{horizonBound}) and (\ref{infinityBound}) imply that $Q_T(\pm\infty) = 0$. Next, we compute
\begin{align*}
-\left(Q_T\right)' &= \omega_I\left|u'\right|^2 + \text{Im}\left(\left(\omega^2-V\right)\overline{\omega}\right)\left|u\right|^2\\
                   & = \omega_I\left(\left|u'\right|^2 + \left(\left|\omega\right|^2 + \frac{(r-2M)\left(rl(l+1) + 2M\right)}{r^4}\right)\left|u\right|^2\right).
\end{align*}
Since the coefficients of $\left|u'\right|^2$ and $\left|u\right|^2$ are positive, the fundamental theorem of calculus implies that $u$ is identically $0$. Algebraically, we are exploiting the fact that the potential $V$ does not depend on $\omega$ and is positive.

Now consider a real mode solution with corresponding $u(r^*)$ and $\omega \in \mathbb{R}\setminus\{0\}$. This time we have ``conservation of energy,''
\[\left(Q_T\right)' = 0.\]
Integrating gives
\[Q_T(\infty)-Q_T(-\infty) = 0 \Rightarrow \]
\[\omega^2\left|u(\infty)\right|^2 + 2Mr_+\omega^2\left|u(-\infty)\right|^2 = 0.\]
We have used the boundary conditions (\ref{horizonBound}) and (\ref{infinityBound}) to evaluate the microlocal energy current at $\pm\infty$. Applying the unique continuation lemma from section \ref{vanish1} immediately implies that $u$ vanishes identically.
\subsection{Modes on Kerr: The Ergoregion, Superradiance, and Whiting's Transformations}
On the Kerr spacetime all of these arguments break down.

In the ergoregion
\[\Delta - a^2\sin^2\theta < 0\]
the Killing vector field $\partial_t$ is no longer timelike. Hence, the associated conserved quantity is no longer coercive and is useless by itself.

At the level of the ODE, we may again define a microlocal energy current:
\[Q_T := \text{Im}\left(u'\overline{\omega u}\right).\]
However,
\[\text{Im}\left(\left(\omega^2-V\right)\overline{\omega}\right) = \]
\[\omega_I\left(\left|\omega\right|^2 - \frac{a^2m^2}{(r^2+a^2)^2} + \frac{\Delta}{(r^2+a^2)^4}\left(a^2\Delta + 2Mr(r^2-a^2)\right)\right) +\] \[\frac{\Delta}{(r^2+a^2)^2}\text{Im}\left(\left(\lambda_{\omega ml}+a^2\omega^2\right)\overline{\omega}\right)\]
is no longer always positive. In fact, for $\omega_I > 0$
\[\text{Im}\left(\left(\omega^2-V\right)\overline{\omega}\right)\left(-\infty\right) = \omega_I\left(\left|\omega\right|^2 - \frac{a^2m^2}{4M^2r_+^2}\right) < 0 \Leftrightarrow \]
\[\left|am\right| - 2Mr_+\left|\omega\right| > 0.\]
This troublesome frequency regime also arises if $\omega \in \mathbb{R}\setminus\{0\}$. For such $\omega$ we still have ``conservation of energy,''
\[\left(Q_T\right)' = 0.\]
Integrating and evaluating with the boundary conditions (\ref{horizonBound}) and (\ref{infinityBound}) gives
\begin{prop}(The Microlocal Energy Estimate)\label{energyEst}
\[\omega^2\left|u(\infty)\right|^2 - \omega\left(am-2Mr_+\omega\right)\left|u(-\infty)\right|^2 = 0.\]
\end{prop}
If $\omega\left(am-2Mr_+\omega\right) < 0$, then this gives a successful estimate of the boundary terms $\left|u(-\infty)\right|^2$ and $\left|u(\infty)\right|^2$. However, if
\begin{equation}\label{superradiantFreq}
\omega\left(am-2Mr_+\omega\right) \geq 0,
\end{equation}
then proposition \ref{energyEst} fails to give an estimate for $\left|u(-\infty)\right|^2$ and $\left|u(\infty)\right|^2$. In the case of (\ref{superradiantFreq}) we say that our frequency parameters are \emph{superradiant}. The existence of superradiant frequencies is the phase space manifestation of the fact that the physical space energy flux associated to $\partial_t$ may be negative along the horizon, i.e.~energy can be extracted from a spinning black hole.

Despite these difficulties, in \cite{n1} Whiting was able to give a relatively short proof of mode stability for a wide class of equations on sub-extremal Kerr, including the wave equation $\Box_g\psi = 0$, i.e.~Theorem \ref{modeStability}. By closely examining the structure of $u$'s and $S_{\omega ml}$'s equations, Whiting found (appropriately non-degenerate) integral and differential transformations taking $u$ to $\tilde{u}$ and $S_{\omega ml}$ to $\tilde{S}_{\omega ml}$ such that
\[\tilde{\psi}(t,r,\theta,\phi) := (r^2+a^2)^{-1/2}e^{-i\omega t}e^{im\phi}\tilde{S}_{\omega ml}(\theta)\tilde{u}(r^*(r))\]
satisfied a wave equation $\Box_{\tilde{g}}\tilde{\psi} = 0$ associated to a new metric $\tilde{g}$ for which there was no ergoregion. After this miracle, the proof concluded with a physical space energy argument as in our discussion of Schwarzschild in section \ref{modeSchw}.
\section{The Wronskian Estimate and Proofs of Mode Stability}\label{wronskEst}
In this section we will explain our extension of Whiting's integral transformations and use this to prove Theorems \ref{qMSRA}, \ref{modeStability}, and \ref{realModeStability}.

It turns out to be useful to work with the inhomogeneous version of $R$'s and $u$'s equations:
\begin{equation}\label{inhomoEqn}
\Delta\frac{d}{dr}\left(\Delta\frac{dR}{dr}\right) - \tilde{V}R = \Delta(r^2+a^2)F(r) =: \Delta\hat{F},
\end{equation}
\[\tilde{V} := -(r^2+a^2)^2\omega^2 + 4Mamr\omega - a^2m^2 + \Delta\left(\lambda_{\omega ml} + a^2\omega^2\right).\]
Here $\hat{F}$ is assumed to be a $C^{\infty}$ function compactly supported in $(r_+,\infty)$.

Recalling that $u(r^*) = (r^2+a^2)^{1/2}R(r)$, we have
\begin{equation}\label{inhomoEqn2}
u'' + \left(\omega^2-V\right)u = H,
\end{equation}
\begin{align*}
V := &\frac{4Mram\omega - a^2m^2 + \Delta(\lambda_{\omega ml} + a^2\omega^2)}{(r^2+a^2)^2}\\
&+\frac{\Delta}{(r^2+a^2)^4}\left(a^2\Delta + 2Mr(r^2-a^2)\right),
\end{align*}
\begin{equation}\label{H}
H(r^*) := \frac{\Delta}{(r^2+a^2)^{1/2}}F(r).
\end{equation}

Our starting point is Whiting's integral transformation:
\begin{align}\label{white}
\tilde{u}(x^*) :=&\ (x^2+a^2)^{1/2}(x-r_+)^{-2iM\omega}e^{-i\omega x}\times\\
&\int_{r_+}^{\infty}e^{\frac{2i\omega}{r_+-r_-}(x-r_-)(r-r_-)}(r-r_-)^{\eta}(r-r_+)^{\xi}e^{-i\omega r}R(r)dr.
\end{align}
Here $\eta$ and $\xi$ are given by
\[\eta := \frac{-i(am-2Mr_-\omega)}{r_+-r_-},\]
\[\xi := \frac{i(am-2Mr_+\omega)}{r_+-r_-}.\]

In \cite{n1} Whiting used the above transformation only on modes satisfying the homogeneous equation with Im$(\omega) > 0$, and the integral was thus absolutely convergent. Since we shall also allow $\omega \in \mathbb{R}\setminus\{0\}$, at first, $\tilde{u}$ only makes sense as an $L^2_{\text{loc}}$ function. Nevertheless, in section \ref{intTransProof} we will establish

\begin{prop}\label{eqn}Let $\text{Im}\left(\omega\right) \geq 0$, $\omega \neq 0$, $R$ solve the inhomogeneous radial ODE (\ref{inhomoEqn}), and $R$ satisfy the boundary conditions from definition \ref{mode}. Define $\tilde{u}$ via Whiting's integral transformation (\ref{white}). Then $\tilde{u}(x)$ is $C^{\infty}$ on $(r_+,\infty)$ and, letting primes denote $x^*$-derivatives, satisfies
\[\tilde{u}'' + \Phi \tilde{u} = \tilde{H},\]
where
\begin{equation}\label{tildeH}
\tilde{H}(x^*) := \frac{(x-r_+)(x-r_-)}{(x^2+a^2)^2}\tilde{G}(x),
\end{equation}
\begin{align*}
\tilde{G}(x) :=&\ (x^2+a^2)^{1/2}(x-r_+)^{-2iM\omega}e^{-i\omega x}\times\\
&\int_{r_+}^{\infty}e^{\frac{2i\omega}{r_+-r_-}(x-r_-)(r-r_-)}(r-r_-)^{\eta}(r-r_+)^{\xi}e^{-i\omega r}\hat{F}(r)dr,
\end{align*}
\[\Phi(x^*) := \frac{(x-r_-)\tilde{\Phi}_1(x)}{(x^2+a^2)^2} - \tilde{\Phi}_2(x),\footnote{For mode stability on the real axis, it is only important that $\Phi$ is real.}\]
\[\tilde{\Phi}_1(x) := \omega^2(x-r_+)^2(x-r_-) - \left(4M\omega^2 + \frac{4\omega(am-2Mr_+\omega)}{r_+-r_-}\right)(x-r_-)(x-r_+) \] \[+4M^2\omega^2(x-r_-) + \left(2am\omega - \lambda_{\omega ml} - a^2\omega^2\right)(x-r_+),\]
\[\tilde{\Phi}_2(x) := \frac{(x-r_+)(x-r_-)}{(x^2+a^2)^4}\left(a^2(x-r_+)(x-r_-) + 2Mx(x^2-a^2)\right).\]
\end{prop}

Of course, it is important to understand the boundary conditions for $\tilde{u}$. When $\text{Im}\left(\omega\right) > 0$, the following quite crude analysis of $\tilde{u}$ is sufficient.
\begin{prop}\label{asymptoticsEasy}If $\text{Im}\left(\omega\right) > 0$, then
\begin{enumerate}
    \item $\tilde{u} = O\left(\left(x-r_+\right)^{2M\text{Im}\left(\omega\right)}\right)$ as $x \to r_+$.
    \item $\tilde{u}' = O\left(\left(x-r_+\right)^{2M\text{Im}\left(\omega\right)}\right)$ as $x \to r_+$.
    \item $\tilde{u} = O\left(e^{-\text{Im}\left(\omega\right)x}x^{1+2M\text{Im}\left(\omega\right)}\right)$ as $x \to \infty$.
    \item $\tilde{u}' = O\left(e^{-\text{Im}\left(\omega\right)x}x^{1+2M\text{Im}\left(\omega\right)}\right)$ as $x \to \infty$.
\end{enumerate}
\end{prop}
When $\omega \in \mathbb{R}\setminus\{0\}$ we need to be a little more precise.
\begin{prop}\label{asymptotics}If $\omega \in \mathbb{R}\setminus\{0\}$, then
\begin{enumerate}
    \item $\tilde{u}$ is uniformly bounded.
    \item $\left|\tilde{u}(\infty)\right|^2 = \frac{(r_+-r_-)^2\left|\Gamma\left(2\xi+1\right)\right|^2}{8M\omega^2 r_+}\left|u(-\infty)\right|^2$.
    \item $\tilde{u}'$ is uniformly bounded.
    \item $\tilde{u}' - i\omega \tilde{u} = O(x^{-1})$ at $x^* = \infty$.
    \item $\tilde{u}' + \frac{i\omega(r_+-r_-)}{r_+}\tilde{u} = O(x-r_+)$ at $x^* = -\infty$.
\end{enumerate}
\end{prop}
Here
\[\Gamma\left(z\right) := \int_0^{\infty}e^{-t}t^{z-1}dt\]
is the Gamma function. Recall the well known fact that the (extended) Gamma function is meromorphic, never vanishes, and only has poles at $0$, $-1$, $-2$, $\cdots$.

Let's see how these propositions restricted to the homogeneous case allow for immediate proofs of both mode stability in the upper half plane and on the real axis via the microlocal energy current:
\[\tilde{Q}_T := \text{Im}\left(\tilde{u}'\overline{\omega\tilde{u}}\right).\]
\begin{proof}(Mode Stability, Theorem \ref{modeStability}) Suppose we had a mode solution with corresponding $(u,S_{\omega ml},\lambda_{\omega ml})$ and $\omega = \omega_R + i\omega_I$ with $\omega_I > 0$. Let $\tilde{u}$ be defined by (\ref{white}). Proposition \ref{asymptoticsEasy} implies that $\tilde{Q}_T\left(\pm\infty\right) = 0$. We proceed as in our discussion of Schwarzschild from section \ref{modeSchw} with $\tilde{u}$ replacing $u$:
\[0 = -\tilde{Q}_T\big|^{\infty}_{-\infty} = -\int_{-\infty}^{\infty}\left(\tilde{Q}_T\right)'dr^* = \int_{-\infty}^{\infty}\left(\omega_I\left|\tilde{u}'\right|^2 + \text{Im}\left(\Phi\overline{\omega}\right)\left|\tilde{u}\right|^2\right)dr^*.\]
Hence, if we can show that $\text{Im}\left(\Phi\overline{\omega}\right) \geq 0$, we may conclude that $\tilde{u}$ vanishes.

An easy computation using the formula from proposition \ref{eqn} gives
\[\text{Im}\left(\Phi\overline{\omega}\right) = \omega_I\left(\frac{(x-r_-)}{(x^2+a^2)^2}\Psi_0 +  \frac{(x-r_+)(x-r_-)}{(x^2+a^2)^4}\Psi_1\right) - \] \[\frac{(x-r_+)(x-r_-)}{(x^2+a^2)^2}\text{Im}\left(\left(\lambda_{\omega ml}+a^2\omega^2\right)\overline{\omega}\right),\]
\[\Psi_0 := \left|\omega\right|^2(x-r_+)^2(x-r_-) + \frac{8M^2\left|\omega\right|^2}{r_+-r_-}(x-r_-)(x-r_+) + 4M^2\left|\omega\right|^2(x-r_-),\]
\[\Psi_1 := a^2(x-r_+)(x-r_-) + 2Mx(x^2-a^2).\]
All of these terms are clearly positive except for $-\text{Im}\left(\left(\lambda_{\omega ml}+a^2\omega^2\right)\overline{\omega}\right)$. For this term we need to return to $S_{\omega ml}$'s equation (\ref{sml}):
\[\frac{1}{\sin\theta}\frac{d}{d\theta}\left(\sin\theta\frac{dS_{\omega ml}}{d\theta}\right) - \left(\frac{m^2}{\sin^2\theta}+a^2\omega^2\sin^2\theta\right)S_{\omega ml} + \left(\lambda_{\omega ml}+a^2\omega^2\right)S_{\omega ml} = 0.\]
Now multiply the equation by $\overline{\omega S_{\omega ml}}\sin\theta$, integrate by parts, and take the imaginary part. There are no boundary terms due to $S_{\omega ml}$'s boundary conditions,\footnote{Recall that the boundary conditions (\ref{boundSml}) required that $e^{im\phi}S_{\omega ml}(\theta)$ extend smoothly to $\mathbb{S}^2$. More explicitly, let $x := \cos\theta$; then an asymptotic analysis of the angular ODE shows that the boundary condition (\ref{boundSml}) is equivalent to $S_{\omega ml} \sim (x\pm 1)^{|m|/2}\text{ as }x \to \mp 1$.} and we find
\[\omega_I\int_0^{\pi}\left(\left|\frac{dS_{\omega ml}}{d\theta}\right|^2 + \left(\frac{m^2}{\sin^2\theta} + a^2\left|\omega\right|^2\sin^2\theta\right)\left|S_{\omega ml}\right|^2\right)\sin\theta d\theta = \] \[-\int_0^{\pi}\left(\text{Im}\left(\left(\lambda_{\omega ml}+a^2\omega^2\right)\overline{\omega}\right)\right)\left|S_{\omega ml}\right|^2\sin\theta d\theta \Rightarrow \]
\[-\text{Im}\left(\left(\lambda_{\omega ml}+a^2\omega^2\right)\overline{\omega}\right) \geq 0.\]
We conclude that $\text{Im}\left(\Phi\overline{\omega}\right)$ is positive, and hence that $\tilde{u}$ must vanish.

In terms of $R$, this implies that
\[\tilde{R}(x) := \int_{r_+}^{\infty}e^{\frac{2i\omega}{r_+-r_-}(x-r_-)(r-r_-)}(r-r_-)^{\eta}(r-r_+)^{\xi}e^{-i\omega r}R(r)dr\]
vanishes for all $x \in (r_+,\infty)$. To see that this implies that $R$ vanishes, we first extend $R$ by $0$ to all of $\mathbb{R}$ and note that the Fourier transform of $(r-r_-)^{\eta}(r-r_+)^{\xi}e^{-i\omega r}R(r)$ is, up to a change of variables,
\[\hat{R}(z) := \int_{-\infty}^{\infty}e^{2i\left|\omega\right|^2z(r-r_-)}(r-r_-)^{\eta}(r-r_+)^{\xi}e^{-i\omega r}R(r)dr.\]
In view of the support of $R$, $\hat{R}$ extends to a holomorphic function on the upper half plane. The vanishing of $\tilde{R}$ for $x \in (r_+,\infty)$ implies that $\hat{R}$ vanishes along the line $\{\frac{y}{\overline{\omega}} : y \in (1,\infty)\}$. Analyticity implies that $\hat{R}$ and hence $R$ itself vanishes.
\end{proof}
Note that the above proof occurs completely at the level of $\tilde{u}$ and $S_{\omega ml}$. In particular, we neither need Whiting's differential transformations of $S_{\omega ml}$ (see section IV of \cite{n1}) nor a physical space argument with a new metric (see section VI of \cite{n1}).
\begin{proof}(Mode Stability on the Real Axis, Theorem \ref{realModeStability})
Suppose we have a mode solution with corresponding $(u,S_{\omega ml},\lambda_{\omega ml})$ and $\omega \in \mathbb{R}\setminus\{0\}$. Let $\tilde{u}$ be defined by (\ref{white}). Then, noting that $\Phi$ from proposition \ref{eqn} is real, we have conservation of energy:
\[\left(\tilde{Q}_T\right)' = 0 \Rightarrow \]
\[\tilde{Q}_T(\infty) - \tilde{Q}_T(-\infty) = 0.\]
Now the boundary conditions from proposition \ref{asymptotics} imply that we get a useful estimate out of this:
\[\tilde{Q}_T(\infty) - \tilde{Q}_T(-\infty) = \]
\[\frac{1}{2}\left(\omega^2|\tilde{u}(\infty)|^2 + |\tilde{u}'(\infty)|^2 + \omega^2\frac{r_+-r_-}{r_+}|\tilde{u}(-\infty)|^2 + \frac{r_+}{r_+-r_-}|\tilde{u}'(-\infty)|^2\right).\]
The unique continuation lemma from section \ref{vanish1} implies that $\tilde{u}$ must vanish.

In terms of $R$, we see that
\[\tilde{R}(y) := \int_{-\infty}^{\infty}e^{2i\omega y(r-r_-)}(r-r_-)^{\eta}(r-r_+)^{\xi}e^{-i\omega r}R(r)dr\]
vanishes for $y \in (1,\infty)$, where we have extended $R$ by $0$ so that it is defined on all of $\mathbb{R}$. However, it is well known that the Fourier transform of a non-trivial function supported in $(0,\infty)$ cannot vanish on an open set.\footnote{This follows from holomorphically extending to the upper half plane and the Schwarz reflection principle.}

As an alternative to this argument, one may instead use the fact from proposition \ref{asymptotics} that
\[\left|\tilde{u}(\infty)\right|^2 = \frac{(r_+-r_-)^2\left|\Gamma\left(2\xi+1\right)\right|^2}{8M\omega^2 r_+}\left|u(-\infty)\right|^2\]
to conclude that $u(-\infty)$ must vanish. Proposition \ref{energyEst} then implies that $u(\infty)$ and hence $u$ vanishes (again using the unique continuation lemma from section \ref{vanish1}).
\end{proof}
Note that this proof is even simpler than the proof of mode stability in the upper half plane since we only need to refer to $\tilde{u}$.

To produce quantitative estimates for the Wronskian we shall need to work a little harder than we did for the qualitative statements. Before proving Theorem \ref{qMSRA} let's recall some notation and prove two propositions and a lemma. Let $\mathscr{A}$ be as in the statement \ref{qMSRA}, let $(\omega,m,l) \in \mathscr{A}$, and $u$ solve (\ref{inhomoEqn2}) with a non-zero, smooth, compactly supported right hand side (\ref{H}). Define $\tilde{u}$ and $\tilde{H}$ via (\ref{white}) and (\ref{tildeH}). Then we have
\begin{prop}\label{firstBoundEstimate}For $(\omega,m,l) \in \mathscr{A}$, $u$ solving satisfying (\ref{inhomoEqn2}) with a smooth, compactly supported right hand side (\ref{H}), and $\epsilon > 0$ we have
\[\left|u(-\infty)\right|^2 \lesssim (4\epsilon)^{-1}\int_{r_+}^{\infty}\left|F(r)\right|^2r^4dr + \epsilon\int_{r_+}^{\infty}\left|R(r)\right|^2dr.\]
\end{prop}
Remark: The implied constants in our $\lesssim$'s will be allowed to depend on the frequency parameters; however, the dependence will always be ``quantitative'' in the sense of theorem \ref{qMSRA}.
\begin{proof}We have
\[\left(\tilde{Q}_T\right)' = \omega\text{Im}\left(\tilde{H}\overline{\tilde{u}}\right) \Rightarrow \]
\[\tilde{Q}_T(\infty) - \tilde{Q}_T(-\infty) = \omega\int_{-\infty}^{\infty}\text{Im}\left(\tilde{H}\overline{\tilde{u}}\right)dx^*.\]
As above, the boundary conditions from proposition \ref{asymptotics} imply that we get a useful estimate:
\[\tilde{Q}_T(\infty) - \tilde{Q}_T(-\infty) = \]
\[\frac{1}{2}\left(\omega^2|\tilde{u}(\infty)|^2 + |\tilde{u}'(\infty)|^2 + \omega^2\frac{r_+-r_-}{r_+}|\tilde{u}(-\infty)|^2 + \frac{r_+}{r_+-r_-}|\tilde{u}'(-\infty)|^2\right).\]
For any $\epsilon > 0$ changing variables and applying Plancherel implies
\[\int_{-\infty}^{\infty}\text{Im}\left(\tilde{H}\overline{\tilde{u}}\right)dr^* \lesssim (4\epsilon)^{-1}\int_{r_+}^{\infty}\left|F(r)\right|^2r^4dr + \epsilon\int_{r_+}^{\infty}\left|R(r)\right|^2dr.\]
To conclude the proof we simply recall that proposition \ref{asymptotics} gives
\[\left|\tilde{u}(\infty)\right|^2 = \frac{(r_+-r_-)^2\left|\Gamma\left(2\xi+1\right)\right|^2}{8M\omega^2 r_+}\left|u(-\infty)\right|^2.\]
\end{proof}

Next, we would like to bootstrap this estimate by working directly with $u$'s/$R$'s ODE to estimate
\[\int_{r_+}^{\infty}\left|R(r)\right|^2dr\]
and then obtain
\begin{prop}\label{secondBoundEstimate}For $(\omega,m,l) \in \mathscr{A}$ and $u$ solving satisfying (\ref{inhomoEqn2}) with a smooth, compactly supported right hand side (\ref{H}), we have
\[\left|u(-\infty)\right|^2 \lesssim \int_{r_+}^{\infty}\left|F(r)\right|^2r^4dr.\]
\end{prop}
Remark: It is important to observe that there are too many powers of $r$ on the right hand side for the above proposition to be directly useful for Theorem \ref{boundMicroLocalEnergy}.
\begin{proof}Following \cite{n2} and \cite{n4}, the ODE techniques used in the proof of this proposition have become fairly standard, e.g. see \cite{n33} and \cite{n36}; hence, in order to focus on the main new ideas we have placed the proof in section \ref{odeEstimates}.
\end{proof}
Next, we switch gears a little and directly construct solutions to the inhomogeneous radial ODE via the following lemma.
\begin{lemm}\label{solFormula}Let $H(x^*)$ be compactly supported. For any $(\omega,m,l) \in \mathscr{A}$, define
\begin{align*}
u(r^*) :=\ &W^{-1}\Bigg(u_{\text{out}}(r^*)\int_{-\infty}^{r^*}u_{\text{hor}}(x^*)H(x^*)dx^* \\
&+ u_{\text{hor}}(r^*)\int_{r^*}^{\infty}u_{out}(x^*)H(x^*)dx^*\Bigg).
\end{align*}
Then
\[u'' + \left(\omega^2-V\right)u = H,\]
and $u$ satisfies the boundary conditions of a mode solution (\ref{horizonBound}) and (\ref{infinityBound}).
\end{lemm}
\begin{proof}This is a simple computation.
\end{proof}
Finally, we can prove Theorem \ref{qMSRA}.
\begin{proof}
Define $\tilde{u}$ via lemma \ref{solFormula}. Then we have
\[\left|u(-\infty)\right|^2 = \left|W\right|^{-2}\left|\int_{-\infty}^{\infty}u_{\text{out}}(x^*)H(x^*)dx^*\right|^2.\]
Combining this with proposition \ref{secondBoundEstimate} gives
\[\left|W\right|^{-2} \lesssim \frac{\int_{r_+}^{\infty}\left|(r^2+a^2)^{1/2}\Delta^{-1}H\right|^2r^4dr}{\left|\int_{-\infty}^{\infty}u_{out}(x^*)H(x^*)dx^*\right|^2}\]
Of course, $W$ is independent of $H$, so it remains to pick any particular compactly supported $H$ we want so that the right hand side is finite. Since for sufficiently large $x$, $\left|u_{\text{out}}-e^{i\omega x^*}\right| \leq \frac{C}{x}$ for an explicit constant $C$ (appendix \ref{odeFacts}), it is certainly possible to find such an $H$. Thus, we have produced a quantitative bound for $W^{-1}$.
\end{proof}
\section{Proof of the Energy Flux Bound and Integrated Local Energy Decay}\label{wronskianReduct}
In this section we shall show that Theorem \ref{qMSRA} (quantitative mode stability on the real axis) implies Theorem \ref{boundMicroLocalEnergy} (boundedness of the energy flux and integrated local energy decay in the bounded-frequency regime).
\subsection{Some Exponential Damping, Boundary Conditions, and a Representation Formula}\label{theBoundary}
We shall use the notation introduced for the statement of Theorem \ref{boundMicroLocalEnergy}. In order to avoid dealing with certain technical issues near null infinity, it turns out to be easier for the proof to work with
\[\psi_{\epsilon} := e^{-\epsilon t}\psi\text{ for }\epsilon \geq 0.\]
Recall that before the statement of Theorem \ref{boundMicroLocalEnergy} we defined a cutoff $\chi$ such that $\chi$ is $0$ in the past of $\Sigma_{0}$ and identically $1$ in the future of $\Sigma_1$. We then define
\[\psi_{\epsilon,\text{\Rightscissors}} := \chi\psi_{\epsilon},\]
\[E_{\epsilon} := e^{-\epsilon t}\left(\left(\Box_g\chi\right)\psi + 2\nabla^{\mu}\chi\nabla_{\mu}\psi\right),\]
\[\omega_{\epsilon} := \omega + i\epsilon.\]
Next, we let $F_{\epsilon}$ be the projection onto the oblate spheroidal harmonics of the Fourier transform of $(r^2+a^2)^{-1}\rho^2E_{\epsilon}$, i.e.
\[F_{\epsilon} := \int_0^{\pi}\int_0^{2\pi}\int_{-\infty}^{\infty}(r^2+a^2)^{-1}\rho^2E_{\epsilon}e^{i\omega t}e^{-im\phi}S_{\omega ml}\sin\theta\ dt\ d\phi\ d\theta.\]
Then let $u_{\epsilon}(r^*)$ similarly be the projection onto the oblate spheroidal harmonics of the Fourier transform of $(r^2+a^2)^{1/2}\psi_{\epsilon}$, and
\[H_{\epsilon}(r^*) = \frac{\Delta}{(r^2+a^2)^{1/2}}F_{\epsilon}.\]
We get
\begin{equation}
u_{\epsilon}'' + \left(\omega_{\epsilon}^2-V_{\epsilon}\right)u_{\epsilon} = H_{\epsilon},
\end{equation}
\begin{align*}
V :=\ &\frac{4Mram\omega_{\epsilon} - a^2m^2 + \Delta(\lambda_{\omega_{\epsilon} ml} + a^2\omega_{\epsilon}^2)}{(r^2+a^2)^2}\\
&+ \frac{\Delta}{(r^2+a^2)^4}\left(a^2\Delta + 2Mr(r^2-a^2)\right).
\end{align*}

For notational ease, we shall introduce one last set of definitions. Recalling the notations established in definition \ref{theHor} and \ref{theOut}, we set
\[u_{\text{hor},\epsilon}(r^*) := u_{\text{hor}}(r^*,\omega_{\epsilon},m,l),\]
\[u_{\text{out},\epsilon}(r^*) := u_{\text{out}}(r^*,\omega_{\epsilon},m,l),\]
\[W_{\epsilon} := u_{\text{out},\epsilon}'u_{\text{hor},\epsilon} - u_{\text{hor},\epsilon}'u_{\text{out},\epsilon}.\]
These will satisfy
    \begin{enumerate}
        \item $u_{\text{hor},\epsilon}'' + \left(\omega_{\epsilon}^2-V_{\epsilon}\right)u_{\text{hor},\epsilon} = 0$.
        \item $u_{\text{hor},\epsilon} \sim (r-r_+)^{\frac{i(am-2Mr_+\omega_{\epsilon})}{r_+-r_-}}\text{ near }r^* = -\infty$.
        \item $\left|\left((r(\cdot)-r_+)^{\frac{-i(am-2Mr_+\omega_{\epsilon})}{r_+-r_-}}u_{\text{hor},\epsilon}\right)\left(-\infty\right)\right|^2 = 1$.
        \item $u_{\text{out},\epsilon}'' + \left(\omega_{\epsilon}^2-V_{\epsilon}\right)u_{\text{out},\epsilon} = 0$.
        \item $u_{\text{out},\epsilon} \sim e^{i\omega_\epsilon r^*}\text{ near }r^* = \infty$
        \item $\left|\left(e^{-i\omega_{\epsilon} \left(\cdot\right)}u_{\text{out},\epsilon}\right)\left(\infty\right)\right|^2 = 1$.
        \item $W_{\epsilon} \neq 0$\text{ by mode stability.}
    \end{enumerate}

The following representation formula is a useful starting point.
\begin{prop}\label{aFormula}For every $\epsilon > 0$,
\begin{align}\label{theFormula}
u_{\epsilon}(r^*) = W_{\epsilon}^{-1}\Bigg(&u_{\text{out},\epsilon}(r^*)\int_{-\infty}^{r^*}u_{\text{hor},\epsilon}(x^*)H_{\epsilon}(x^*)dx^*\\
&+ u_{\text{hor},\epsilon}(r^*)\int_{r^*}^{\infty}u_{\text{out},\epsilon}(x^*)H_{\epsilon}(x^*)dx^*\Bigg).
\end{align}
\end{prop}
\begin{proof}As explained in appendix \ref{admissibleBounded}, standard arguments show that the admissibility assumption implies that $\left\vert\left\vert \chi\psi\right\vert\right\vert_{L^{\infty}} =: \beta < \infty$. Hence, we will have
\begin{equation}\label{expDecay}
\left|\psi_{\epsilon}\right| \lesssim \beta e^{-\epsilon t}.
\end{equation}

Along the support of $\psi_{\epsilon}$, there exists a constant $B$ such that $t \geq \left|r^*\right| - B$. We conclude that $\psi_{\epsilon}$ in fact satisfies
\begin{equation}\label{expDecay2}
\left|\psi_{\epsilon}\right| \lesssim \beta \exp\left(-\frac{\epsilon\left( \left|r^*\right| + t\right)}{2}\right).
\end{equation}
It is easy to see from this that $u_{\epsilon}$ is exponentially decreasing as $r^* \to \pm\infty$ (remember that $\epsilon > 0$). Since the argument of appendix \ref{admissibleBounded} also applies to the derivatives of $\psi$, we may also conclude that $H_{\epsilon}$ is exponentially decreasing as $r^* \to \pm\infty$. Hence, we can define
\begin{align*}
\hat{u}_{\epsilon}(r^*) := W_{\epsilon}^{-1}\Bigg(&u_{\text{out},\epsilon}(r^*)\int_{-\infty}^{r^*}u_{\text{hor},\epsilon}(x^*)H_{\epsilon}(x^*)dx^*\\
&+ u_{\text{hor},\epsilon}(r^*)\int_{r^*}^{\infty}u_{\text{out},\epsilon}(x^*)H_{\epsilon}(x^*)dx^*\Bigg).
\end{align*}

Now, a simple computation shows that
\[\left(\hat{u}_{\epsilon}-u_{\epsilon}\right)'' + \left(\omega_{\epsilon}^2-V\right)\left(\hat{u}_{\epsilon}-u_{\epsilon}\right) = 0.\]
Furthermore, $\hat{u}_{\epsilon}-u_{\epsilon}$ is exponentially decreasing as $r^* \to \pm\infty$. From ODE theory (appendix \ref{odeFacts}), $\hat{u}_{\epsilon}-u_{\epsilon}$ must be asymptotic to a linear combination of
\[\left\{(r-r_+)^{\frac{i(am-2Mr_+\omega_{\epsilon})}{r_+-r_-}},(r-r_+)^{\frac{-i(am-2Mr_+\omega_{\epsilon})}{r_+-r_-}}\right\}.\]
The only possible choice is
\[\hat{u}_{\epsilon}-u_{\epsilon} \sim (r-r_+)^{\frac{i(am-2Mr_+\omega_{\epsilon})}{r_+-r_-}}\text{ at }r^* = -\infty.\]
Next, ODE theory (appendix \ref{odeFacts}) implies that near infinity, $\hat{u}_{\epsilon}-u_{\epsilon}$ must be asymptotic to a linear combination of
\[\left\{e^{i\omega_{\epsilon} r^*},e^{-i\omega_{\epsilon} r^*}\right\}.\]
The exponential decay of $\hat{u}_{\epsilon}-u_{\epsilon}$ singles out
\[\hat{u}_{\epsilon}-u_{\epsilon} \sim e^{i\omega_{\epsilon} r^*}\text{ at }r^* = \infty.\]
Thus, $\hat{u}_{\epsilon}-u_{\epsilon}$ satisfies the boundary conditions of a mode solution. Finally, mode stability in the upper half plane implies that $\hat{u}_{\epsilon} = u_{\epsilon}$.
\end{proof}
It will be convenient to use the above formula when $\epsilon = 0$; however, it must be understood in an $L^2$ sense. First we need two lemmas.
\begin{lemm}\label{errorEnergy}
\[\limsup_{\epsilon\to 0}\int_{\mathscr{B}}\sum_{m,l}\int_{r_+}^{\infty}\left|F_{\epsilon}\right|^2r^2\, dr\, d\omega = \int_{\mathscr{B}}\sum_{m,l}\int_{r_+}^{\infty}\left|F\right|^2r^2\, dr\, d\omega \lesssim \int_{\Sigma_0}\left|\partial\psi\right|^2\]
\end{lemm}
In particular, even though there are $0$th order terms in $F$, there are only derivatives of $\psi$ on the right hand side.
\begin{proof}By Plancherel,
\[\limsup_{\epsilon\to 0}\int_{\mathscr{B}}\sum_{m,l}\int_{r_+}^{\infty}\left|F_{\epsilon}\right|^2r^2\, dr\, d\omega \lesssim\] \[\limsup_{\epsilon\to 0}\int_{(t,r,\theta,\phi)}\left(\left|e^{-\epsilon t}\left(\Box_g\chi\right)\psi\right|^2 + \left|e^{-\epsilon t}\nabla^{\mu}\chi\nabla_{\mu}\psi\right|^2\right)r^2\sin\theta\, dt\, dr\, d\theta\, d\phi.\]
We will consider the two terms on the right hand side separately.

For the second term, we simply observe that the asymptotic behavior of $\Sigma_0$ implies
\[\left|e^{-\epsilon t}\nabla^{\mu}\chi\nabla_{\mu}\psi\right|^2 \lesssim \] \[e^{-2\epsilon t}1_{\text{supp}\left(\nabla\chi\right)}\Big(\left|\left(\partial_t+\partial_{r^*}\right)\psi\right|^2 + O(r^{-2})\left|\left(\partial_t-\partial_{r^*}\right)\psi\right|^2 +\] \[O(r^{-2})\left(\left|\partial_{\theta}\psi\right|^2 + \left|\partial_{\phi}\psi\right|^2\right)\Big)\]
where $1_{\text{supp}\left(\nabla\chi\right)}$ denotes the indicator function on the support of $\nabla\chi$.

For the first term, first pick a null frame $(L,\underline{L},E_1,E_2)$ where
\[g(L,L) = g(\underline{L},\underline{L}) = g(E_1,E_2) = g(L,E_i) = g(\underline{L},E_i) = 0,\]
\[g(L,\underline{L}) = -2,\]
\[g(E_1,E_1) = g(E_2,E_2) = 1,\]
\[L = \partial_t+\partial_{r^*} + O(r^{-1}),\]
\[\underline{L} = \partial_t - \partial_{r^*} + O(r^{-1}).\]
Expanding $\Box_g$ in this null frame (see \cite{n8}) gives
\[\left|\Box_g\chi\right|^2 = \left|-\underline{L}L\chi + E^2_1\chi + E^2_2\chi + \left(\nabla_{\underline{L}}L - \nabla_{E_1}E_1 - \nabla_{E_2}E_2\right)\chi\right|^2 \lesssim\]
\[1_{\text{supp}\left(\nabla\chi\right)}r^{-2}.\]
In summary, we have
\[\left|e^{-\epsilon t}\left(\Box_g\chi\right)\psi\right|^2 + \left|e^{-\epsilon t}\nabla^{\mu}\chi\nabla_{\mu}\psi\right|^2 \lesssim 1_{\text{supp}\left(\nabla\chi\right)}\left(\frac{\left|e^{-\epsilon t}\psi\right|^2}{r^2} + e^{-2\epsilon t}\left|\partial\psi\right|^2\right) \Rightarrow \]
\[\limsup_{\epsilon\to 0}\int_{\mathscr{B}}\sum_{m,l}\int_{r_+}^{\infty}\left|F_{\epsilon}\right|^2r^2dr\, d\omega \lesssim\] \[\limsup_{\epsilon\to 0}\int_{(t,r,\theta,\phi)}1_{\text{supp}\left(\nabla\chi\right)}\left(\frac{\left|e^{-\epsilon t}\psi\right|^2}{r^2} + e^{-2\epsilon t}\left|\partial\psi\right|^2\right)r^2\sin\theta dt\, dr\, d\theta\, d\phi \lesssim \]
\begin{equation}\label{stuff2}
\int_{(t,r,\theta,\phi)}1_{\text{supp}\left(\nabla\chi\right)}\left|\partial\psi\right|^2r^2\sin\theta dt\, dr\, d\theta\, d\phi.
\end{equation}

In the last inequality we used a Poincar\'e inequality to control the $0$th order term (note that due to the support of $\nabla\chi$, for each $t$, the $r$ integration occurs over a region of bounded size). Finally, keeping in mind that $\nabla\chi$ is only supported in between $\Sigma_0$ and $\Sigma_1$, we observe that (\ref{stuff2}) is controlled by a constant times
\[\sup_{s \in [0,1]}\int_{\Sigma_s}\left|\partial\psi\right|^2 \lesssim \int_{\Sigma_0}\left|\partial\psi\right|^2.\]
The last inequality uses a \emph{finite-in-time} non-degenerate energy estimate (see \cite{n8}).
\end{proof}
\begin{lemm}\label{cancellation}
\[\left\vert\left\vert\int_{r^*}^{\infty}u_{\text{hor}}(x^*)H(x^*)dx^*\right\vert\right\vert_{L^2_{\omega\in\mathscr{B},(m,l)\in\mathscr{C}}} \leq B\left(r^*,C_{\mathscr{B}},C_{\mathscr{C}}\right) \int_{\Sigma_0}\left|\partial\psi\right|^2,\]

and
\[\lim_{\epsilon\to 0}\int_{r^*}^{\infty}u_{\text{hor},\epsilon}(x^*)H_{\epsilon}(x^*)dx^* =\int_{r^*}^{\infty}u_{\text{hor}}(x^*)H(x^*)dx^* \text{ in }L^2_{\omega\in\mathscr{B},(m,l)\in\mathscr{C}}.\]
\end{lemm}
\begin{proof}
We start with the first assertion. Note that a naive application of Cauchy-Schwarz followed by Plancherel would produce too many powers of $x^*$; however, if we somehow gained a power of $x^{-1}$ we could always use the inequality
\[\left|\int_{r^*}^{\infty}u_{\text{out}}(x^*)H(x^*)x^{-1}dx^*\right|^2 \lesssim \int_{r\left(r^*\right)}^{\infty}\left|F\right|^2r^2dr.\]
After integrating in $\omega$ and summing in $(m,l)$, this can be controlled by Lemma \ref{errorEnergy}. We will denote by $G$ all terms that can be controlled by this sort of brute force Cauchy-Schwarz inequality. Let's return to the troublesome term. We start by observing that
\[\int_{\mathscr{B}}\sum_{(m,l)\in\mathscr{C}}\left|\int_{r^*}^{\infty}u_{\text{out}}(x^*)H(x^*)dx^*\right|^2d\omega = \] \[\int_{\mathscr{B}}\sum_{(m,l)\in\mathscr{C}}\left(\left|\int_{r^*}^{\infty}e^{i\omega x^*}H(x^*)dx^*\right|^2+ G\right)d\omega.\]
The plan is to take advantage of the oscillations in $\omega$ by a suitable application of Plancherel. However, we will first need to account for all of the $\omega$ dependence in $H$. Let's introduce the variables
\[u := \frac{1}{2}(t-r^*)\]
\[v := \frac{1}{2}(t+r^*).\]
From the definitions of the cutoff and the triangle inequality, it follows that
\[\left|\int_{r^*}^{\infty}e^{i\omega x^*}H(x^*)dx^*\right|^2 \lesssim\]
\[\left|\int_{\mathbb{S}^2}\int_{-\infty}^{\infty}\int_{A}^{\infty}e^{2i\omega v}\left(\partial_u\chi\right)\left(\partial_v\psi\right)e^{-im\phi}S_{\omega ml}(\theta,\omega)r\sin\theta dv\, du\, d\theta\,  d\phi\right|^2 + \]
\[\left|\int_{\mathbb{S}^2}\int_{-\infty}^{\infty}\int_{A}^{\infty}e^{2i\omega v}\left(\Box_g\chi\right)\psi e^{-im\phi}S_{\omega ml}(\theta,\omega)r\sin\theta dv\, du\, d\theta\, d\phi\right|^2+ G\]
Here $A$ denotes a large fixed constant possibly depending on $r^*$. Let's focus on the first term on the right hand side since the second term will be treated similarly. Using Plancherel relative to the orthonormal basis $\{e^{im\phi}S_{\omega ml}(\theta)\}$ of $L^2(\sin\theta d\theta d\phi)$ gives
\[\int_{\mathscr{B}}\sum_{(m,l)\in\mathscr{C}}\Bigg|\int_{\mathbb{S}^2}\int_{-\infty}^{\infty}\int_{A}^{\infty}e^{2i\omega v}\left(\partial_u\chi\right)\left(\partial_v\psi\right)e^{-im\phi}\times\]
\[S_{\omega ml}(\theta,\omega)r\sin\theta dv\, du\, d\theta\, d\phi\Bigg|^2d\omega  \]
\begin{equation}\label{blah}
\lesssim\int_{\mathbb{S}^2}\int_{\mathscr{B}}\left|\int_{-\infty}^{\infty}\int_{A}^{\infty}e^{2i\omega v}\left(\partial_u\chi\right)\left(\partial_v\psi\right)rdv\ du\right|^2d\omega\, \sin\theta d\theta\, d\phi.
\end{equation}
Due to the support of $\partial_u\chi$, the $u$ integration occurs over a region of uniformly bounded size. Hence, Cauchy-Schwarz in the $u$ integral implies that (\ref{blah}) is controlled by
\[\int_{\mathbb{S}^2}\int_{-\infty}^{\infty}\int_{-\infty}^{\infty}\left|\int_{A}^{\infty}e^{2i\omega v}\partial_u\chi\partial_v\psi rdv\right|^2d\omega\, du\, \sin\theta d\theta\, d\phi  \lesssim \]
\[\int_{\mathbb{S}^2}\int_{-\infty}^{\infty}\int_{A}^{\infty}\left|\partial_u\chi\partial_v\psi\right|^2 r^2 du\, dv\, \sin\theta d\theta\, d\phi\]
Now we can just appeal to (the proof of) Lemma \ref{errorEnergy}. For the term
\[\left|\int_{\mathbb{S}^2}\int_{-\infty}^{\infty}\int_{A}^{\infty}e^{2i\omega v}\left(\Box_g\chi\right)\psi e^{-im\phi}S_{\omega ml}(\theta,\omega)r\sin\theta dv\, du\, d\theta\, d\phi\right|^2\]
we can carry out exactly the same procedure except that we add a Poincar\'e inequality (just as in proposition \ref{errorEnergy}) at the end so that we can close the estimate at the level of derivatives of $\psi$. In conclusion, we have
\[\int_{\mathscr{B}}\sum_{(m,l)\in\mathscr{C}}\left|\int_{A}^{\infty}u_{\text{out}}(x^*)H(x^*)dx^*\right|^2d\omega \lesssim \int_{\Sigma_0}\left|\partial\psi\right|^2.\]

It is now clear that the second assertion of the lemma can be proved by essentially repeating the above argument with the difference of \[\int_{r^*}^{\infty}u_{\text{hor}}(x^*)H(x^*)dx^*\]
and
\[\int_{r^*}^{\infty}u_{\text{hor},\epsilon}(x^*)H_{\epsilon}(x^*)dx^*.\]
\end{proof}
Now we are ready to prove the following.
\begin{lemm}\label{niceLemma}Let $\mathscr{B} \subset \mathbb{R}$ and $\mathscr{C} \subset \{(m,l) \in \mathbb{Z}\times \mathbb{Z} : l \geq |m|\}$ be such that
\[C_{\mathscr{B}} := \sup_{\omega \in \mathscr{B}}\left(\left|\omega\right| + \left|\omega\right|^{-1}\right) < \infty,\]
\[C_{\mathscr{C}} := \sup_{m,l \in \mathscr{C}}\left(\left|m\right| + \left|l\right|\right) < \infty.\]
Then, for each $r^* \in (-\infty,\infty)$, the formula
\begin{align}\label{superFormula}
u(r^*) = W^{-1}\Bigg(&u_{\text{out}}(r^*)\int_{-\infty}^{r^*}u_{\text{hor}}(x^*)H(x^*)dx^*\\
&+ u_{\text{hor}}(r^*)\int_{r^*}^{\infty}u_{\text{out}}(x^*)H(x^*)dx^*\Bigg).
\end{align}
holds in $L^2_{\omega \in \mathscr{B}}l^2_{(m,l) \in \mathscr{C}}$.
\end{lemm}
\begin{proof}We start with the formula (\ref{theFormula}) and justify the $\epsilon \to 0$ limit term by term. Of course, the convergence of $u_{\epsilon}$ to $u$ is simply a consequence of Plancherel. Next, we observe the following facts (see appendix \ref{odeFacts}):
\begin{enumerate}
    \item For any $A_0 > -\infty$, $u_{\text{out},\epsilon} \to u_{\text{out}}$ in $L^{\infty}_{r^* \in [A_0,\infty),\omega \in \mathscr{B},(m,l)\in\mathscr{C}}$.
    \item For any $A_1 < \infty$, $u_{\text{hor},\epsilon} \to u_{\text{hor}}$ in $L^{\infty}_{r^* \in (-\infty,A_1],\omega \in \mathscr{B},(m,l)\in\mathscr{C}}$.
    \item $W_{\epsilon}^{-1} \to W^{-1}$ in $L^{\infty}_{\omega \in \mathscr{B},(m,l) \in \mathscr{C}}$.
\end{enumerate}
Thus, it suffices to restrict attention to the two integrals. The term
\[\int_{r^*}^{\infty}u_{\text{out}}(x^*)H(x^*)dx^*\]
has already been treated in lemma \ref{cancellation}. For the other term, it is sufficient to oberserve
\[\int_{\omega \in \mathscr{B}}\sum_{(m,l) \in \mathscr{C}}\left|\int_{-\infty}^{r^*}u_{\text{hor}}(x^*)H(x^*)dx^*\right|^2\ d\omega = \]
\[\int_{\omega \in \mathscr{B}}\sum_{(m,l) \in \mathscr{C}}\left|\int_{r_+}^{r\left(r^*\right)}u_{\text{hor}}(x^*\left(x\right))F\left(x\right)\left(x^2+a^2\right)^{1/2}dx\right|^2\ d\omega \lesssim_{r^*} \]
\begin{equation}\label{finalThing}
\int_{\omega \in \mathscr{B}}\sum_{(m,l) \in \mathscr{C}}\int_{r_+}^{r\left(r^*\right)}\left|F\left(x\right)\right|^2dx\ d\omega.
\end{equation}
We have used the facts
\[H(r^*) = \Delta(r^2+a^2)^{-1/2}F(r),\]
\[dr^* = \frac{(r^2+a^2)}{\Delta}dr.\]
We may control (\ref{finalThing}) via lemma \ref{errorEnergy}. Given this, it is easy to justify the limit as $\epsilon \to 0$.
\end{proof}
In the same fashion, one may prove the following.
\begin{lemm}\label{moreFormulas}
\begin{align*}
u'(r^*) = W^{-1}\Bigg(&u'_{\text{out}}(r^*)\int_{-\infty}^{r^*}u_{\text{hor}}(x^*)H(x^*)dx^*\\
&+ u'_{\text{hor}}(r^*)\int_{r^*}^{\infty}u_{\text{out}}(x^*)H(x^*)dx^*\Bigg),
\end{align*}
and
\[\lim_{r^*\to\infty}\left(u' - i\omega u\right) = 0.\]
Both equalities are understood in the same way as in lemma \ref{niceLemma}.
\end{lemm}

\subsection{The Estimate}
We keep the notation introduced in the previous section. Also, recall the definition of $\left|\partial \psi\right|^2$ given in the statement of Theorem \ref{boundMicroLocalEnergy} and appendix \ref{theFirstt}.

We now prove Theorem \ref{boundMicroLocalEnergy}.
\begin{proof}By Plancherel, it suffices to prove
\[\int_{\mathscr{B}}\sum_{(m,l) \in \mathscr{C}}\left(\left(\left|u(-\infty)\right|^2 +  \left|u(\infty)\right|^2\right) + \int_{r_0}^{r_1}\left(\left|u'\right|^2+ \left|u\right|^2\right)dr^*\right)d\omega \leq \]
\[B\left(r_0,r_1,C_{\mathscr{B}},C_{\mathscr{C}}\right) \int_{\Sigma_0}\left|\partial\psi\right|^2.\]

We begin with (\ref{superFormula}) which gives
\begin{align}\label{usefulEqn}
u(r^*) = W^{-1}\Bigg(&u_{\text{out}}(r^*)\int_{-\infty}^{r^*}u_{\text{hor}}(x^*)H(x^*)dx^*\\
&+ u_{\text{hor}}(r^*)\int_{r^*}^{\infty}u_{\text{out}}(x^*)H_(x^*)dx^*\Bigg).
\end{align}
The equality is in $L^2_{\omega \in\mathscr{B}}l^2_{(m,l)\in\mathscr{C}}$.

The following properties are simple consequences of the construction of $u_{\text{out}}$ and $u_{\text{hor}}$ (see appendix \ref{odeFacts}).
\begin{enumerate}
    \item $\left\vert\left\vert u_{\text{out}}\right\vert\right\vert_{L^{\infty}_{r^*,\omega \in \mathscr{B},(m,l)\in\mathscr{C}}} < \infty$.
    \item $\left\vert\left\vert u_{\text{hor}}\right\vert\right\vert_{L^{\infty}_{r^*,\omega \in \mathscr{B},(m,l)\in\mathscr{C}}} < \infty$.
\end{enumerate}

Thus, simply evaluating (\ref{usefulEqn}) at $r^* = -\infty$, integrating, and summing,  gives
\begin{equation}\label{eqn1}
\int_{\mathscr{B}}\sum_{(m,l)\in\mathscr{C}}\left|u(-\infty)\right|^2\ d\omega \leq
\end{equation}
\[\limsup_{r^*\to-\infty} \int_{\mathscr{B}}\sum_{(m,l)\in\mathscr{C}}W^{-2}\left|\int_{r^*}^{\infty}u_{\text{out}}(x^*)H(x^*)dx^*\right|^2\ d\omega.\footnote{The point being that one may easily show \[\lim_{r^*\to -\infty}\int_{\mathcal{B}}\sum_{\mathcal{C}}\left|\int_{-\infty}^{r^*}u_{\text{hor}}\left(x^*\right)H\left(x^*\right)dx^*\right|^2 = 0.\]}\]

Next, for $A$ much larger than $r_1$, we have
\begin{equation}\label{anLinftyEst}
\int_{\mathscr{B}}\sum_{(m,l)\in\mathscr{C}}\left\vert\left\vert u\right\vert\right\vert^2_{L^{\infty}\left(r_0,r_1\right)}\ d\omega \lesssim
\end{equation}
\begin{align*}
\int_{\mathscr{B}}\sum_{(m,l)\in\mathscr{C}} W^{-2}\Bigg(&\sup_{r^*\in\left[r_0,r_1\right]}\left|\int_{-\infty}^{r^*}u_{\text{hor}}(x^*)H(x^*)dx^*\right|^2\\ &+\sup_{r^* \in\left[r_0,r_1\right]}\left|\int_{r^*}^Au_{\text{out}}(x^*)H(x^*)dx^*\right|^2\\
&+\left|\int_A^{\infty}u_{\text{out}}(x^*)H(x^*)dx^*\right|^2\Bigg)\ d\omega.
\end{align*}

We have already used multiple times that
\begin{align*}
\sup_{r^*\in\left[r_0,r_1\right]}\left|\int_{-\infty}^{r^*}u_{\text{hor}}(x^*)H(x^*)dx^*\right|^2 &\lesssim \left(\int_{-\infty}^{r_1}\left|H\right|dx^*\right)^2\\
&= \left(\int_{r_+}^{r_1}\left|F(r)(r^2+a^2)^{1/2}\right|dr\right)^2\\
&\lesssim\int_{r_+}^{r_1}\left|F\right|^2dr.
\end{align*}
The constant will depend on $r_1$, but that does not concern us. Combining (the proof of) this estimate with (\ref{anLinftyEst}) gives
\begin{equation}\label{anLinftyEst2}
\int_{\mathscr{B}}\sum_{(m,l)\in\mathscr{C}}\left\vert\left\vert u\right\vert\right\vert^2_{L^{\infty}\left(r_0,r_1\right)}\ d\omega \lesssim
\end{equation}
\[\int_{\mathscr{B}}\sum_{(m,l)\in\mathscr{C}}\left(W^{-2}\left(\int_{r_+}^{r(A)}\left|F\right|^2dr + \left|\int_A^{\infty}u_{\text{out}}(x^*)H(x^*)dx^*\right|^2\right)\right)\ d\omega.\]

Of course, we may integrate this $L^{\infty}$ estimate to get
\begin{equation}\label{eqn2}
\int_{\mathscr{B}}\sum_{(m,l)\in\mathscr{C}}\int_{r_0}^{r_1}\left|u\right|^2dr^*\ d\omega \lesssim
\end{equation}
\[\int_{\mathscr{B}}\sum_{(m,l)\in\mathscr{C}}\left(W^{-2}\left(\int_{r_+}^{r(A)}\left|F\right|^2dr + \left|\int_A^{\infty}u_{\text{out}}(x^*)H(x^*)dx^*\right|^2\right)\right)\ d\omega.\]

Next, via lemma \ref{moreFormulas}, we may essentially differentiate (\ref{usefulEqn}), and proceed exactly as the proof of (\ref{eqn2}) in order to establish
\begin{equation}\label{eqn3}
\int_{r_0}^{r_1}\left|u'\right|^2dr^* \lesssim
\end{equation}
\[\int_{\mathscr{B}}\sum_{(m,l)\in\mathscr{C}}\left(W^{-2}\left(\int_{r_+}^{r(A)}\left|F\right|^2dr + \left|\int_A^{\infty}u_{\text{out}}(x^*)H(x^*)dx^*\right|^2\right)\right)\ d\omega.\]

Lastly, to control $\left|u(\infty)\right|^2$, we use the already introduced microlocal energy current and lemma \ref{moreFormulas} to conclude
\[\omega^2\left|u(\infty)\right|^2 = Q_T(\infty) = Q_T(-\infty) + \int_{-\infty}^{\infty}\left(Q_T\right)'dr^* \Rightarrow  \]
\begin{equation}\label{eqn4}
\int_{\mathscr{B}}\sum_{(m,l)\in\mathscr{C}}\left|u(\infty)\right|^2\ d\omega\lesssim
\end{equation}
\[\int_{\mathscr{B}}\sum_{(m,l)\in\mathscr{C}}\left(\omega(am-2Mr_+\omega)\left|u(-\infty)\right|^2 + \omega\int_{-\infty}^{\infty}\text{Im}\left(H\overline{u}\right)dr^*\right)\ d\omega.\]

After applying Plancherel, the proof of lemma \ref{errorEnergy}, lemma \ref{cancellation}, Theorem \ref{qMSRA}, and adding inequalities (\ref{eqn1}), (\ref{eqn2}), (\ref{eqn3}), and (\ref{eqn4}) together,  we get
\[\int_{\mathscr{B}}\sum_{(m,l) \in \mathscr{C}}\left(\left|u(-\infty)\right|^2 +  \left|u(\infty)\right|^2\right)d\omega + \int_{\mathscr{B}}\sum_{(m,l)\in\mathscr{C}}\int_{r_0}^{r_1}\left(\left|u'\right|^2+ \left|u\right|^2\right)dr^*\, d\omega \]
\[\lesssim  \int_{\Sigma_0}\left|\partial\psi\right|^2.\]
\end{proof}
Before concluding the section, we would like to emphasize that for the applications to \cite{n2}, it is crucial that we have arranged for the right hand side of this estimate be given by a non-degenerate energy flux through $\Sigma_0$.
\section{The Integral Transformation}\label{intTransProof}
In this section we will prove propositions \ref{eqn}, \ref{asymptoticsEasy}, and \ref{asymptotics}. For clarity of exposition we will restrict ourselves to $\omega \in \mathbb{R}\setminus\{0\}$; indeed, for $\text{Im}\left(\omega\right) > 0$ the proofs are much easier and follow from the same sort of reasoning as the real $\omega$ case. Furthermore, due to the symmetries of the radial ODE, we may restrict ourselves to $\omega > 0$.

It will be convenient to adopt the notation
\[A := \frac{2i\omega}{r_+-r_-}.\]
It will also be useful to consider the following functions
\[g(r) := (r-r_+)^{-\xi}(r-r_-)^{-\eta}e^{i\omega r}R(r),\]
\[\tilde{g}(z) := \int_{r_+}^{\infty}e^{A(z-r_-)(r-r_-)}(r-r_-)^{2\eta}(r-r_+)^{2\xi}e^{-2i\omega r}g(r)dr\]
\begin{equation}\label{theDef}
= \int_{r_+}^{\infty}e^{A(z-r_-)(r-r_-)}(r-r_-)^{\eta}(r-r_+)^{\xi}e^{-i\omega r}R(r)dr.
\end{equation}
Here $z = x+iy$ with $y \geq 0$. Recall the previously defined
\[\eta := \frac{-i(am-2Mr_-\omega)}{r_+-r_-},\]
\[\xi := \frac{i(am-2Mr_+\omega)}{r_+-r_-}.\]

If $y > 0$ then the integrals and their derivatives are all absolutely convergent; we immediately conclude that $\tilde{g}$ is holomorphic for $z$ in the upper half plane. When $y = 0$, then $\tilde{g}(x)$ is, a priori, only an $L^2$ function; however, in section \ref{defineRealAxis} we will show that $\tilde{g}\left(x\right)$ is in fact a $C^1$ function on $[r_+,\infty)$. Then, in section \ref{theNewEqn} we will verify $\tilde{g}$'s equation and show that $\tilde{g}$ is smooth on $(r_+,\infty)$. Finally, in section \ref{asymptoticAnalysis} we will carry out an asymptotic analysis of $\tilde{g}\left(x\right)$ as $x\to\infty$; in particular, we will identify $\lim_{x\to\infty}\left|x\tilde{g}(x)\right|$. Putting everything together will prove propositions \ref{eqn}, \ref{asymptoticsEasy}, and \ref{asymptotics}.
\subsection{Defining $\tilde{g}$ on the Real Axis}\label{defineRealAxis}
For any $y > 0$ and $\epsilon > 0$, we shall rewrite $\tilde{g}$ in the following way:
\begin{lemm}\label{intByPartsEst}

\[\tilde{g}(z) = \int_{r_+}^{r_++\epsilon}e^{A(z-r_-)(r-r_-)}(r-r_-)^{\eta}(r-r_+)^{\xi}e^{-i\omega r}R(r)dr\]
\begin{align*}
- \Bigg(&\left(A(z-r_-)\right)^{-1}e^{A(z-r_-)(r_+-r_-+\epsilon)}(r_+-r_-+\epsilon)^{\eta}\\
&\times \epsilon^{2\xi}e^{-i\omega(r_++\epsilon)}\epsilon^{-\xi}R(r_++\epsilon)\Bigg)
\end{align*}
\begin{align*}
+\Bigg(&\left(A(z-r_-)\right)^{-2}e^{A(z-r_-)(r_+-r_-+\epsilon)}\\
&\times\frac{d}{dr}\left((\cdot-r_-)^{\eta}(\cdot-r_+)^{\xi}e^{-i\omega \cdot}R(\cdot)\right)(r_++\epsilon)\Bigg)
\end{align*}
\begin{align*}
+\left(A(z-r_-)\right)^{-2}\int_{r_++\epsilon}^{\infty}\Bigg(&e^{A(z-r_-)(r-r_-)}\\
&\frac{d^2}{dr^2}\left((r-r_-)^{\eta}(r-r_+)^{\xi}e^{-i\omega r}R(r)\right)\Bigg)dr.
\end{align*}
\end{lemm}
\begin{proof}This follows by integrating by parts twice the expression (\ref{theDef}) in a straightforward manner.
\end{proof}
\begin{lemm}\label{regularg}The function $\tilde{g}(x)$ is continuous on $[r_+,\infty)$ and $O\left(x^{-1}\right)$ as $x\to\infty$.
\end{lemm}
\begin{proof}Recall that the boundary conditions for $R$, (\ref{horizonBound}) and (\ref{infinityBound}), imply
\begin{enumerate}
    \item $(r-r_+)^{-\xi}R(r)$ is smooth at $r_+$.
    \item $\frac{d^k}{dr^k}\left(e^{-i\omega r}R(r)\right) = O\left(r^{-k-1}\right)\text{ as }r\to\infty$.
\end{enumerate}
In particular, the integral in the last line of the formula from (\ref{intByPartsEst}) is absolutely convergent even when $y = 0$. Thus, even when $y = 0$, we may conclude that the right hand side of the formula is continuous in $x$.

In order to see the decay in $x$, set $\epsilon = x^{-1}$. By direct inspection one finds that each term is $O(x^{-1})$. Since the right hand side of the formula converges in $L^2$ as $y \downarrow 0$, by uniqueness of $L^2$ limits we conclude that $\tilde{g}(x)$ is equal to the formula. The lemma then follows.
\end{proof}

Now we turn to $\frac{\partial g}{\partial x}$. We have
\begin{lemm}\label{theDerivativeFormula}For any $y > 0$ and $\epsilon > 0$ we have
\[\frac{\partial \tilde{g}}{\partial x} - A(r_+-r_-)\tilde{g} = \]
\[-(z-r_-)^{-1}\int_{r_+}^{r_++\epsilon}e^{A(z-r_-)(r-r_-)}\frac{d}{dr}\left((r-r_-)^{\eta}(r-r_+)^{\xi+1}e^{-i\omega r}R(r)\right)dr + \]
\begin{align*}
\Big(&A^{-1}(z-r_-)^{-2}e^{A(z-r_-)(r_+-r_-+\epsilon)}\\
&\times\frac{d}{dr}\left((r-r_-)^{\eta}(r-r_+)^{\xi+1}e^{-i\omega r}R(r)\right)(r_++\epsilon)\Big)
\end{align*}
\begin{align*}
-\Big(&A^{-2}(z-r_-)^{-3}e^{A(z-r_-)(r_+-r_-+\epsilon)}\\
&\times\frac{d^2}{dr^2}\left((r-r_-)^{\eta}(r-r_+)^{\xi+1}e^{-i\omega r}R(r)\right)(r_++\epsilon)\Big)
\end{align*}
\begin{align*}
-A^{-2}(z-r_-)^{-3}\int_{r_++\epsilon}^{\infty}&e^{A(z-r_-)(r-r_-)}\\
&\times\frac{d^3}{dr^3}\left((r-r_-)^{\eta}(r-r_+)^{\xi+1}e^{-i\omega r}R(r)\right)dr.
\end{align*}
\end{lemm}
\begin{proof}This follows from a straightforward series of integration by parts on the expression
\[\frac{\partial \tilde{g}}{\partial x} - A(r_+-r_-)\tilde{g} = \]
\[A\int_{r_+}^{\infty}e^{A(z-r_-)(r-r_-)}(r-r_-)^{\eta}(r-r_+)^{\xi+1}e^{-i\omega r}R(r)dr.\]
\end{proof}
Next, we have
\begin{lemm}$\frac{\partial \tilde{g}}{\partial x}(x)$ exists and is continuous on $[r_+,\infty)$. Furthermore
\[\frac{\partial \tilde{g}}{\partial x} - A\left(r_+-r_-\right)\tilde{g} = O\left(x^{-2}\right)\text{ as }x\to\infty.\]
\end{lemm}
\begin{proof}This follows by setting $\epsilon = x^{-1}$ in lemma \ref{theDerivativeFormula} and then reasoning as in lemma \ref{regularg}.
\end{proof}

\subsection{Verifying the New Equation}\label{theNewEqn}
In this section we will compute $\tilde{g}$'s new equation.

We say that a function $h$ satisfies a Confluent Heun Equation (CHE) if there are complex parameters $\gamma$, $\delta$, $p$, $\alpha$, and $\sigma$ and a function $G$ such that
\begin{equation}\label{CHEop}
Th := (r-r_+)(r-r_-)\frac{d^2h}{dr^2} + \left(\gamma(r-r_+) + \delta(r-r_-) + p(r-r_+)(r-r_-)\right)\frac{dh}{dr} +
\end{equation}
\[\left(\alpha p(r-r_-) + \sigma\right)h = G.\]
One finds that $g$ satisfies such a CHE with
\[\gamma = 2\eta + 1 =: \gamma_0,\]
\[\delta = 2\xi + 1 =: \delta_0,\]
\[p = -2i\omega =: p_0,\]
\[\alpha = 1 =:\alpha_0,\]
\[\sigma = 2am\omega - 2\omega r_- i - \lambda_{\omega ml} - a^2\omega^2 =: \sigma_0,\]
\[G = (r-r_+)^{-\xi}(r-r_-)^{-\eta}e^{i\omega r}\hat{F} =: G_0.\]
We need an integration by parts lemma whose straightforward proof is omitted.
\begin{lemm}\label{intByParts}Let $T$ denote a Confluent Heun operator as defined in (\ref{CHEop}). Then
\[\int_{\beta_1}^{\beta_2}\left(Tf\right)(r-r_+)^{\delta-1}(r-r_-)^{\gamma-1}e^{pr}hdr\]
\[= (r-r_+)^{\delta}(r-r_-)^{\gamma}e^{pr}\left(\frac{df}{dr}h-f\frac{dh}{dr}\right)\bigg|_{\beta_1}^{\beta_2} \]
\[+ \int_{\beta_1}^{\beta_2}\left(Th\right)(r-r_+)^{\delta-1}(r-r_-)^{\gamma-1}e^{pr}fdr.\]
\end{lemm}
Next we will compute $\tilde{g}$'s equation for $y > 0$.
\begin{lemm}\label{theEqn0}If $y > 0$ we have
\[(z-r_+)(z-r_-)\frac{\partial^2\tilde{g}}{\partial x^2} + \]
\[\left((z-r_+) + (1-4iM\omega)(z-r_-) - 2i\omega(z-r_-)(z-r_+)\right)\frac{\partial \tilde{g}}{\partial x} +\]
\[\left(-2i\omega(2\eta+1)(z-r_-) + 2am\omega - 2\omega r_-i-\lambda_{\omega ml}-a^2\omega^2\right)\tilde{g} = \tilde{G}\]
where
\[\tilde{G} := \int_{r_+}^{\infty}e^{\frac{2i\omega}{r_+-r_-}(z-r_-)(r-r_-)}(r-r_-)^{2\eta}(r-r_+)^{2\xi}e^{-2i\omega r}G_0(r)dr.\]
\end{lemm}
\begin{proof}
Since the coefficients of the CHE are all holomorphic, we may take the derivatives in the CHE to be complex derivatives. Let $L_r$ denote a Confluent Heun Operator in the $r$ variable with parameters $\left(\gamma_0,\delta_0,p_0,\alpha_0,\sigma_0\right)$ and right hand side $G_0$. Let $\tilde{L}_z$ denote a Confluent Heun operator in the $z$ $(=x+iy)$ variable with, to be determined, tilded parameters.

We wish to determine if
\[\int_{r_+}^{\infty}e^{A(z-r_-)(r-r_-)}(r-r_-)^{2\eta}(r-r_+)^{2\xi}e^{-2i\omega r}g(r)dr\]
is a solution to a CHE with tilded parameters. When $y > 0$ the exponential damping in the integral allows differentiation under the integral sign, and we see from Lemma \ref{intByParts} that the following two conditions will suffice:
\[\left(\tilde{L}_z - L_r\right)e^{A(z-r_-)(r-r_-)} = 0,\]
\[(r-r_+)^{\delta_0}(r-r_-)^{\gamma_0}e^{p_0r}e^{A(z-r_-)(r-r_-)}\left(A(z-r_-)g - \frac{dg}{dr}\right)\bigg|_{r_+}^{\infty} = 0\]
\[ \forall z\text{ such that }y > 0.\]
We have
\[e^{-A(z-r_-)(r-r_-)}\left(\tilde{L}_z - L_r\right)e^{A(z-r_-)(r-r_-)} = \]
\[A\left(A(r_+-r_-) + \tilde{p}\right)(r-r_-)(z-r_-)^2 -A\left(A(r_+-r_-)+p_0\right)(r-r_-)^2(z-r_-)\]
\[-A\left(\gamma_0+\delta_0+p_0(r_+-r_-)-\tilde{\gamma}-\tilde{\delta}-\tilde{p}(r_+-r_-)\right)(z-r_-)(r-r_-)\]
\[+\left(A\gamma(r_+-r_-)+\tilde{\alpha}\tilde{p}\right)(z-r_-) -\left(A\tilde{\gamma}(r_+-r_-) + \alpha_0 p_0\right)(r-r_-) + \left(\tilde{\sigma}-\sigma_0\right).\]
From this it is clear that we must have
\[A = -p(r_+-r_-)^{-1} = 2i\omega(r_+-r_-)^{-1},\]
\[\tilde{p} = p_0 = -2i\omega,\]
\[\tilde{\alpha} = \gamma_0,\]
\[\tilde{\gamma} = \alpha_0 = 1,\]
\[\tilde{\delta} = \gamma_0 + \delta_0 - \tilde{\gamma} = 1 - 4iM\omega,\]
\[\tilde{\sigma} = \sigma_0.\]
We still need to check that the boundary conditions are satisfied. Since $g$ and $\frac{dg}{dr}$ both decay for large $r$, the exponential decay from $e^{A(z-r_-)(r-r_-)}$ clearly implies that
\[\left((r-r_+)^{\delta_0}(r-r_-)^{\gamma_0}e^{p_0r}e^{A(z-r_-)(r-r_-)}\left(A(z-r_-)g - \frac{dg}{dr}\right)\right)(r = \infty) = 0\]
\[\text{ for all }z\text{ with }y > 0.\]
Since $\delta_0 = 2\xi + 1$, with $\xi$ purely imaginary, and $|g|$ extends continuously to $r_+$, we see that
\[\left((r-r_+)^{\delta_0}(r-r_-)^{\gamma_0}e^{p_0r}e^{A(z-r_-)(r-r_-)}\left(A(z-r_-)g - \frac{dg}{dr}\right)\right)(r = r_+) = 0\Leftrightarrow\]
\[\frac{dg}{dr^*}(r_+) = 0.\]
If we $r^*$ differentiate the expression defining $g$, we get
\[ \left|\frac{dg}{dr^*}\right|(r_+) = \left|\frac{dR}{dr^*}-\frac{\xi(r_+-r_-)}{2Mr_+}R\right|(r_+) = 0.\]
We conclude that $\tilde{g}$ satisfies $\tilde{L}_z\tilde{g} = 0$. Lastly, since $\tilde{g}$ is holomorphic in the upper half plane, $\frac{d\tilde{g}}{dz} = \frac{\partial\tilde{g}}{\partial x}$.
\end{proof}
Finally, using the analysis from section \ref{defineRealAxis} we can upgrade this lemma to
\begin{lemm}\label{theEqn}When $y = 0$, $\tilde{g}$ is smooth in $(r_+,\infty)$ and we have
\begin{equation}\label{star}
(x-r_+)(x-r_-)\frac{\partial^2\tilde{g}}{\partial x^2} +
\end{equation}
\[\left((x-r_+) + (1-4iM\omega)(x-r_-) - 2i\omega(x-r_-)(x-r_+)\right)\frac{\partial \tilde{g}}{\partial x} +\]
\[\left(-2i\omega(2\eta+1)(x-r_-) + 2am\omega - 2\omega r_-i-\lambda_{\omega ml}-a^2\omega^2\right)\tilde{g} = \tilde{G}\]
where
\[\tilde{G} := \int_{r_+}^{\infty}e^{\frac{2i\omega}{r_+-r_-}(x-r_-)(r-r_-)}(r-r_-)^{2\eta}(r-r_+)^{2\xi}e^{-2i\omega r}G_0(r)dr.\]
\end{lemm}
\begin{proof}One consequence of the analysis in section \ref{defineRealAxis} is that $\tilde{g}\left(x +iy\right)$ converges to $\tilde{g}$ in $H_x^1$ as $y \to 0$. In particular, we may take $y \to 0$ in the weak formulation of the equation from lemma \ref{theEqn0} to conclude that $\tilde{g}(x)$ is a weak $H^1_x$ solution of (\ref{star}). Since $\tilde{G}$ is smooth,\footnote{Recall that \[\tilde{G}(x) =  \int_{r_+}^{\infty}e^{\frac{2i\omega}{r_+-r_-}(x-r_-)(r-r_-)}(r-r_-)^{\eta}(r-r_+)^{\xi}e^{-i\omega r}\hat{F}dr\]
where $\hat{F}$ is smooth and compactly supported in $(r_+,\infty)$.} we may then conclude the proof by an appeal to elliptic regularity.
\end{proof}
\subsection{Asymptotic Analysis of $\tilde{g}$}\label{asymptoticAnalysis}
Recall that in section \ref{defineRealAxis} we saw that $\tilde{g} = O\left(x^{-1}\right)$ as $x\to\infty$. In this section we will carry out the somewhat subtle task of identifying
\[\lim_{x\to\infty}\left|x\tilde{g}(x)\right|.\]

We start with
\begin{lemm}\label{warmUp}Let $h$ be a smooth function on $[r_+,\infty)$ which vanishes on $[r_++2,\infty)$. Recall that we previously defined
\[\xi := \frac{i(am-2Mr_+\omega)}{r_+-r_-} \in i\mathbb{R}.\]
For $\tau \geq 0$ and $\nu > 0$, define
\[Z\left(\nu,\tau\right) := \int_{r_+}^{\infty}e^{i\nu r}\left(r-r_+ + i\tau\right)^{2\xi}h(r)dr.\]
Then we have
\[\left|Z\left(\nu,\tau\right)\right| \lesssim \nu^{-1}\]
where the implied constant does not depend on $\tau$.
\end{lemm}
\begin{proof}
Integrating by parts twice produces the following expression for $Z\left(\nu,\tau\right)$:
\begin{align}
Z\left(\nu,\tau\right) &= \int_{r_+}^{r_++\nu^{-1}}e^{i\nu r}\left(r-r_++i\tau\right)^{2\xi}h(r)dr \\
                 &-\left(i\nu\right)^{-1}e^{i\nu \left(r_++\nu^{-1}\right)}\left(\nu^{-1}+i\tau\right)^{2\xi}h\left(r_++\nu^{-1}\right) \\
                 &+\left(i\nu\right)^{-2}e^{i\nu \left(r_++\nu^{-1}\right)}\frac{d}{dr}\left(\left(\cdot-r_++i\tau\right)^{2\xi}h(\cdot)\right)\left(r_++\nu^{-1}\right)\\
                 &+\left(i\nu\right)^{-2}\int_{r_++\nu^{-1}}^{\infty}e^{i\nu r}\frac{d^2}{dr^2}\left(\left(r-r_++i\tau\right)^{2\xi}h(r)\right)dr.
\end{align}
The lemma follows by direct inspection of each term.
\end{proof}

The following lemma is the technical core of our argument. The proof is a slight adaption from similar problems discussed in the books \cite{n37} and \cite{n38}.
\begin{lemm}\label{technicalLemm}Let $h$ be a smooth function on $[r_+,\infty)$ which vanishes in $[r_++2,\infty)$. Recall that we previously defined
\[\xi := \frac{i(am-2Mr_+\omega)}{r_+-r_-} \in i\mathbb{R}.\]
For $\nu > 0$, define
\[Z\left(\nu\right) := Z\left(\nu,0\right) = \int_{r_+}^{\infty}e^{i\nu r}\left(r-r_+\right)^{2\xi}h(r)dr.\]
Then we have
\[Z\left(\nu\right) = \exp\left(\frac{i\pi}{2}\left(1+2\xi\right)\right)\Gamma\left(2\xi+1\right)h(r_+)e^{i\nu r_+}\nu^{-1-2\xi} + O\left(\nu^{-2}\right)\text{ as }\nu\to\infty\]
where
\[\Gamma(z) := \int_0^{\infty}e^{-t}t^{z-1}dt\]
is the Gamma function.
\end{lemm}
\begin{proof}The key trick is to come up with a clever form of the anti-derivative of $e^{i\nu r}\left(r-r_+\right)^{2\xi}$. In order to do this, we extend $e^{i\nu r}\left(r-r_+\right)^{2\xi}$ to $s \in \mathbb{C}\setminus\{(-\infty,r_+]\}$ where we are taking the principle branch of $(s-r_+)^{2\xi}$. One may easily check that $(s-r_+)^{2\xi} = \exp\left(2\xi\log\left(s-r_+\right)\right)$ is uniformly bounded in the region
\[\left\{s : \text{Re}\left(s\right) \in [r_+,r_++2)\right\}.\]

Thus, keeping in the mind the exponential decay from $e^{i\nu s}$ as $\text{Im}\left(s\right) \to \infty$ and Cauchy's theorem, we may unambiguously define
\[l\left(r,\nu\right) := -\int^{i\infty}_re^{i\nu s}\left(s-r_+\right)^{2\xi}ds\]
whenever $r \in (r_+,r_++2).$ This will satisfy
\[\frac{\partial l}{\partial r} = e^{i\nu r}\left(r-r_+\right)^{2\xi}.\]

Now, integrating along the curve $t \mapsto r + it$ implies
\begin{align}
l\left(r,\nu\right) =\label{goodExpression} -ie^{i\nu r}\int_0^{\infty}e^{-\nu t}\left(r-r_++it\right)^{2\xi}dt.
 \end{align}
Now, keeping in mind that $z^{2\xi} := \exp\left(2\xi\log z\right)$, we have
\begin{align}
\lim_{r \to r_+}l\left(r,\nu\right) &= -i^{1+2\xi}e^{i\nu r_+}\int_0^{\infty}e^{-\nu t}t^{2\xi}dt \\
&= -i^{1+2\xi}e^{i\nu r_+}\nu^{-1-2\xi}\Gamma\left(2\xi+1\right).
\end{align}

More generally, changing variables in \ref{goodExpression} implies
\begin{equation}\label{good2}
l\left(r,\nu\right) = -ie^{i\nu r}\nu^{-1}\int_0^{\infty}e^{-t}\left(r-r_++i\frac{t}{\nu}\right)^{2\xi}dt.
\end{equation}

Now we are ready for an estimate:
\begin{align}
Z\left(\nu,\tau\right) &= \int_{r_+}^{\infty}e^{i\nu r}\left(r-r_+\right)^{2\xi}h(r)dr\\
                  &= \int_{r_+}^{\infty}\frac{\partial l}{\partial r}\left(r,\nu\right)h(r)dr \\
                  &= i^{1+2\xi}\Gamma\left(2\xi+1\right)h(r_+)e^{i\nu r_+}\nu^{-1-2\xi}\\
                  &- \int_{r_+}^{\infty}l\left(r,\nu\right)h'(r)dr\\
                  &= i^{1+2\xi}\Gamma\left(2\xi+1\right)h(r_+)e^{i\nu r_+}\nu^{-1-2\xi}\\
                  &+ i\nu^{-1}\int_0^{\infty}e^{-t}\left(\int_{r_+}^{\infty}e^{i\nu r}\left(r-r_++i\frac{t}{\nu}\right)^{2\xi}h'(r)dr\right)dt.
\end{align}
We have used (\ref{good2}) and Fubini in the last equality.

To conclude the proof we just need to show that
\begin{equation}\label{toConclude}
\int_0^{\infty}e^{-t}\left(\int_{r_+}^{\infty}e^{i\nu r}\left(r-r_++i\frac{t}{\nu}\right)^{2\xi}h'(r)dr\right)dt = O\left(\nu^{-1}\right).
\end{equation}
However, this follows by an application of lemma \ref{warmUp} to the inner integral.
\end{proof}
Let's apply this analysis to $\tilde{g}$.
\begin{lemm}\label{toInf}As $x\to\infty$ we have
\begin{align}
\tilde{g}(x) = \Bigg(&\exp\left(\frac{i\pi}{2}\left(1+2\xi\right)\right)\Gamma\left(2\xi+1\right)\left(r_+-r_-\right)^{\eta}e^{-A\left(r_+-r_-\right)r_-}e^{-i\omega r_+}\\&\times\left(2\omega\left(r_+-r_-\right)^{-1}\right)^{-1-2\xi}\left(\left(\cdot-r_+\right)^{-\xi}R\left(\cdot\right)\right)\left(r_+\right)\\
&\times e^{Ax\left(r_+-r_-\right)}x^{-1-2\xi}\Bigg) + O\left(x^{-2}\right).
\end{align}
\end{lemm}
\begin{proof}Let $\chi(r)$ be a positive smooth function which is identically $1$ on $[r_+,r_++1]$ and identically $0$ on $[r_++2,\infty)$. We may write
\begin{align}
\tilde{g}(x) &= \int_{r_+}^{\infty}e^{A(x-r_-)(r-r_-)}(r-r_-)^{\eta}(r-r_+)^{\xi}e^{-i\omega r}R(r)\chi(r)dr\\
             &+ \int_{r_+}^{\infty}e^{A(x-r_-)(r-r_-)}(r-r_-)^{\eta}(r-r_+)^{\xi}e^{-i\omega r}R(r)\left(1-\chi(r)\right)dr.
\end{align}
The second integral satisfies
\[\int_{r_+}^{\infty}e^{A(x-r_-)(r-r_-)}(r-r_-)^{\eta}(r-r_+)^{\xi}e^{-i\omega r}R(r)\left(1-\chi(r)\right)dr\]
\begin{align*}
=\left(A\left(x-r_-\right)\right)^{-2}\int_{r_+}^{\infty}&e^{A(x-r_-)(r-r_-)}\\
&\times\frac{d^2}{dr^2}\left((r-r_-)^{\eta}(r-r_+)^{\xi}e^{-i\omega r}R(r)\left(1-\chi(r)\right)\right)dr
\end{align*}
\[= O\left(x^{-2}\right).\]
We have used the boundary condition (\ref{infinityBound}).

Now we conclude the proof by applying lemma \ref{technicalLemm} (with $\nu = Ax$) to the first integral.
\end{proof}
\subsection{Putting Everything Together}
Now we will prove propositions \ref{eqn} and \ref{asymptotics}.
\begin{proof}(Proposition \ref{eqn})

Recall the definition of $\tilde{u}$:
\begin{align}
\tilde{u}(x^*) :=&\ (x^2+a^2)^{1/2}(x-r_+)^{-2iM\omega}e^{-i\omega x}\times\\
&\int_{r_+}^{\infty}e^{\frac{2i\omega}{r_+-r_-}(x-r_-)(r-r_-)}(r-r_-)^{\eta}(r-r_+)^{\xi}e^{-i\omega r}R(r)dr.
\end{align}
In terms of $\tilde{g}$ we have
\[\tilde{u}\left(x^*\right) = (x^2+a^2)^{1/2}(x-r_+)^{-2iM\omega}e^{-i\omega x}\tilde{g}(x).\]
In particular $\tilde{u}$ is smooth on $(r_+,\infty)$ and proposition \ref{eqn} follows from lemma \ref{star} and a straightforward (if tedious) calculation.
\end{proof}
\begin{proof}(Proposition \ref{asymptotics})

Keeping in mind that
\[\tilde{u}' = \frac{(x-r_+)(x-r_-)}{x^2+a^2}\frac{\partial\tilde{u}}{\partial x},\]
the lemma follows immediately from
\[\tilde{u}\left(x^*\right) = (x^2+a^2)^{1/2}(x-r_+)^{-2iM\omega}e^{-i\omega x}\tilde{g}(x),\]
the fact that $\tilde{g}$ is $C^1$ at $r_+$ (see section \ref{defineRealAxis}), and lemma \ref{toInf}.
\end{proof}

Recall that we are omitting the proof of proposition \ref{asymptoticsEasy} since it is much easier and follows from the same sort of reasoning as the proofs of propositions \ref{eqn} and \ref{asymptotics}.
\section{Some Estimates for the Kerr ODE}\label{odeEstimates}
For the purposes of section \ref{wronskEst} we need to prove proposition \ref{secondBoundEstimate}:
\begin{equation}\label{theEst}
\left|u(-\infty)\right|^2 \lesssim (4\epsilon)^{-1}\int_{r_+}^{\infty}\left|F(r)\right|^2r^4dr + \epsilon\int_{r_+}^{\infty}\left|R(r)\right|^2dr\ \forall \epsilon > 0 \Rightarrow
\end{equation}
\begin{equation}\label{theEst2}
\left|u(-\infty)\right|^2 \lesssim \int_{r_+}^{\infty}\left|F(r)\right|^2r^4dr.
\end{equation}
It will sometimes be useful to switch our perspectives on $-\infty$ and $\infty$ and write
\[u'' + \left(\omega_0^2 - V_0\right)u = H\]
where
\[\omega_0 = \omega - \frac{am}{2Mr_+},\]
\[V_0 = V + \omega_0^2 - \omega^2.\]
For the following estimates the relevant properties of $V$ and $V_0$ are
\begin{enumerate}
    \item $V$ is uniformly bounded.
    \item $V = O(r^{-2})$ at $\infty$.
    \item $V_0 = O(r-r_+)$.
    \item For fixed non-zero $a$, $m$, and $M > 0$ there exists a constant $c > 0$ such that $am-2Mr_+\omega \geq -c\left(\lambda_{\omega ml}+a^2\omega^2\right) \Rightarrow \frac{dV_0}{dr}(r_+) > 0$.
\end{enumerate}
The last statement is the only non-obvious one, and the relevant computations can be found in \cite{n2}. It will also be useful to note that
\[\lambda_{\omega ml}+a^2\omega^2 \geq |m|\left(|m|+1\right).\]
This follows from the observation that when $a^2\omega^2 = 0$, the $e^{im\phi}S_{\omega ml}(\theta)$ are simply spherical harmonics with corresponding eigenvalues all larger than $|m|(|m|+1)$.

We will explore various estimates and their realm of applicability. Then at the end we will show how they can be combined to establish (\ref{theEst2}). We will borrow the ``separated current template'' from \cite{n2}.
\subsection{Virial Estimate I}\label{virialEstimate}
The estimates of this section require that $\omega$ be bounded away from $0$ and that we have a priori control of $Q_T(\infty)$. The resulting estimate will be sufficiently good near $\infty$, but will require strengthening near $-\infty$.

The virial current is
\[Q^y := y|u'|^2 + y\left(\omega^2-V\right)|u|^2\]
where $y$ is a suitably chosen function.
We have
\[\left(Q^y\right)' = y'|u'|^2 + y'\omega^2|u|^2 - \left(yV\right)'|u|^2 + 2y\text{Re}\left(H\overline{u}'\right).\]
Integrating this gives
\[\int_{-\infty}^{\infty}\left(y'|u'|^2 + y'\omega^2|u|^2 - \left(yV\right)'|u|^2\right)dr^* = \]
\[Q^y(\infty)-Q^y(-\infty) - \int_{-\infty}^{\infty}2y\text{Re}\left(u'\overline{H}\right)dx^*.\]
We want to choose $y$ so that the left hand side controls $|u|^2 + |u'|^2$ (possibly with weights), and so that the boundary terms are controllable. Let $\zeta(r^*)$ be a non-negative function which is identically $1$ near $r^* = -\infty$ and equals $r^{-2}$ near $r^* = \infty$. We set
\[y(r^*) := \exp\left(-B\int_{r^*}^{\infty}\zeta dr^*\right).\]
Here $B$ is a large parameter to be chosen later. We have $y(r_+) = 0$, $y(\infty) = 1$, and $y' = B\zeta y > 0$. We will show that the term
\[-\int_{-\infty}^{\infty}\left(yV\right)'|u|^2dr^*\]
which threatens to destroy the coercivity of our estimate can in fact be absorbed into the other two terms. After an integration by parts and the inequality $|ab| \leq \epsilon|a| + (4\epsilon)^{-1}|b|$, we find
\[\left|\int_{-\infty}^{\infty}\left(yV\right)'|u|^2dr^*\right| \leq \]
\[\frac{1}{2}\int_{-\infty}^{\infty}y'|u'|^2dr^* + 2\int_{-\infty}^{\infty}\left(y'\omega^2\right)\frac{y^2|V|^2}{\omega^2\left(y'\right)^2}|u|^2dr^* + \left|\left(yV|u|^2\right)\big|^{\infty}_{-\infty}\right|.\]
Note that $|V|$ is uniformly bounded, decays like $r^{-2}$, and that $y/y' \leq B^{-1}r^2$. Also, the boundary terms clearly vanish. Thus, for sufficiently large $B$, we get
\[\left|\int_{-\infty}^{\infty}\left(yV\right)'\left|u\right|^2dr^*\right| \leq \frac{1}{2}\int_{-\infty}^{\infty}\left(y'|u'|^2 + y'\omega^2|u|^2\right)dr^*.\]
Lastly, we note that
\[Q^y(\infty) = 2Q_T(\infty).\]
Thus, we end up with
\begin{equation}
\int_{-\infty}^{\infty}\left(y'|u'|^2 + y'\omega^2|u|^2\right)dr^* \lesssim Q_T(\infty) - \int_{-\infty}^{\infty}y\text{Re}\left(u'\overline{H}\right)dr^*.
\end{equation}
The usual Cauchy-Schwarz argument then gives
\begin{equation}\label{strictVirialEstimate}
\int_{-\infty}^{\infty}\left(y'|u'|^2 + y'\omega^2|u|^2\right)dr^* \lesssim \left|Q_T(\infty)\right| + \int_{-\infty}^{\infty}y|H|^2r^2dr^*.
\end{equation}
This estimate is sufficiently strong away from the horizon. However, near $-\infty$, the exponential decay of the weight $y$ makes the estimate quite weak.

\subsection{Virial Estimate II}
In this section we look at the virial current from the opposite direction. This estimate will require that $\omega_0$ is bounded away from $0$ and that we have a priori control of $Q_T(-\infty)$. The resulting estimate will be sufficiently strong near $r_+$, but will require strengthening near $\infty$.

We rewrite the virial current as
\[Q^y := y|u'|^2 + y\left(\omega_0^2-V_0\right)|u|^2.\]
Let $\zeta(r)$ be a positive function equal to $\Delta$ near $r = r_+$, and equal to $1$ near $r = \infty$. Then define
\[y(r^*) := \exp\left(-B\int_{-\infty}^{r^*}\zeta dr^*\right).\]
Integrating the virial current gives
\[\int_{-\infty}^{\infty}\left(-y'|u'|^2 - y'\omega_0^2|u|^2 + \left(yV_0\right)'|u|^2\right)dr^* = \]
\[-Q^y(\infty)+Q^y(-\infty) + \int_{-\infty}^{\infty}2y\text{Re}\left(u'\overline{H}\right)dx^*.\]
We may deal with the $\left(yV_0\right)'$ exactly as in the previous section. This time
\[Q^y(\infty) = 0\]
\[Q^y(-\infty) \approx 2\frac{\omega_0}{\omega} Q_T(-\infty).\]
We end up with
\begin{equation}
\int_{-\infty}^{\infty}\left(-y'|u'|^2 - y'\omega_0^2|u|^2\right)dr^* \lesssim -\frac{\omega_0}{\omega} Q_T(-\infty) + \int_{-\infty}^{\infty}y\text{Re}\left(u'\overline{H}\right)dr^*.
\end{equation}
As in the previous section, it is clear that we also have
\begin{equation}\label{strictVirialEstimate2}
\int_{-\infty}^{\infty}\left(-y'|u'|^2 - y'\omega_0^2|u|^2\right)dr^* \lesssim \left|\frac{\omega_0}{\omega} Q_T(-\infty)\right| + \int_{r_+}^{\infty}\left|F\right|^2dr.
\end{equation}
This estimate is sufficiently strong away from $\infty$. However, near $\infty$, the exponential decay of the weight $y$ makes the estimate very weak.
\subsection{The Red-Shift Estimate}
The estimate of this section will require that $\frac{dV_0}{dr}(r_+) > 0$ and that we can already estimate \[\int_{\alpha}^{\beta}\left(\left|u'\right|^2 + \left|u\right|^2\right)dr^*\]
for arbitrary $r_+ < \alpha < \beta < \infty$.

The following Poincar\'e type inequality will be useful.
\begin{lemm}\label{redPoincare}
Suppose $h$ has support in $[r_+,r_+ + \epsilon]$ and has
\[\left((\cdot-r_+)\left|h\right|^2(\cdot)\right)(r_+) = 0.\]
Then
\[\int_{r_+}^{\infty}|h|^2dr \leq C(\epsilon)\int_{r_+}^{\infty}\left|h' + i\omega_0h\right|^2dr\]
where
\[C(\epsilon) \lesssim \left(1+\epsilon^2\right).\]
\end{lemm}
\begin{proof}Keeping in mind that
\[\frac{dh}{dr^*} = \frac{(r-r_+)(r-r_-)}{r^2+a^2}\frac{dh}{dr},\]
we have
\[\int_{r_+}^{\infty}|h|^2dr = \int_{r_+}^{\infty}\frac{d}{dr}\left(r-r_+\right)|h|^2dr = -\int_{r_+}^{\infty}\left(r-r_+\right)\left(\frac{dh}{dr}\overline{h} + h\frac{d\overline{h}}{dr}\right)dr =\]
\[-\int_{r_+}^{\infty}\left(\frac{r^2+a^2}{r-r_-}\right)\left(h'\overline{h} + h\overline{h}'\right)dr = \] \[-\int_{r_+}^{\infty}\left(\frac{r^2+a^2}{r-r_-}\right)\left(\left(h'+i\omega_0h\right)\overline{h} + h\left(\overline{h}'-i\omega_0\overline{h}\right)\right)dr.\]
From here the lemma follows by the usual argument.
\end{proof}
The (microlocal) red-shift current is
\[Q^z_{\text{red}} := z\left|u'+i\omega_0u\right|^2 - zV_0|u|^2 = Q^z + 2z\frac{\omega_0}{\omega}Q_T.\]
Note that the boundary conditions for $R$ imply that $(u'+i\omega_0u)(r^*) = O(r-r_+)$ near $r^* = -\infty$. Hence, we may take $z$ to be a function which blows up at $-\infty$.
We have
\[\left(Q^z_{\text{red}}\right)' = z'\left|u' + i\omega_0u\right|^2 - \left(zV_0\right)'|u|^2 + 2z\text{Re}\left(\left(u'+i\omega_0u\right)\overline{H}\right).\]
Let $\zeta(r)$ be a bump function identically $1$ on $[r_+,r_++\epsilon]$ and vanishing on $[r_++2\epsilon,\infty)$. $\epsilon$ is a free parameter that we will later take sufficiently small. Now set
\[z(r^*) := -\frac{\zeta(r(r^*))}{V_0}.\]
Note that $z' > 0$ near $-\infty$ since $\frac{d}{dr}V_0(-\infty) > 0$. We have
\[\left(Q^z_{\text{red}}\right)\big|_{-\infty}^{\infty} = -|u(-\infty)|^2\]
which has a good sign. For $r \in [r_+,r_++\epsilon]$, we have
\[\left(Q^z_{\text{red}}\right)' = z'|u'+i\omega_0u|^2 + 2z\text{Re}\left(\left(u'+i\omega_0u\right)\overline{H}\right).\]
Note that we have $z' \sim (r-r_+)^{-1}$ in this region.\footnote{Keep in mind that \[z' = \frac{(r-r_+)(r-r_-)}{r^2+a^2}\frac{dz}{dr}.\]} For $r \in [r_++\epsilon,r_++2\epsilon]$ we will treat everything as an error:
\[\left|\left(Q^z_{\text{red}}\right)\right| \lesssim \left(|u'|^2 + |u|^2\right) + \left|z\text{Re}\left(\left(u'+i\omega_0u\right)\overline{H}\right)\right|.\]
Of course for $r \geq r_+ + 2\epsilon$ we have $\left(Q^z_{\text{red}}\right)' = 0$. Putting everything together will produce an estimate for
\[\int_{r_+}^{r_++\epsilon}(r-r_+)^{-2}|u'(r^*(r))+i\omega_0u(r^*(r))|^2dr.\]
For $\epsilon$ sufficiently small, an application of Lemma \ref{redPoincare} will show that this controls \[\int_{r_+}^{r_++\epsilon/2}|u(r^*(r))|^2dr\]
at the expense of introducing error terms \[\int_{r_++\epsilon/2}^{r_++\epsilon}\left(|u'(r^*(r))|^2+|u(r^*(r))|^2\right)dr.\]
We end up with
\[\int_{r_+}^{r_++\epsilon}(r-r_+)^{-2}|u'(r^*(r))+i\omega_0u(r^*(r))|^2dr + \int_{r_+}^{r_++\epsilon/2}\left|u(r^*(r))\right|^2dr \lesssim \]
\begin{equation}\label{redShiftEstimate}
\int_{\epsilon/2}^{2\epsilon}\left(\left|u\right|^2 + \left|u'\right|^2\right)dr^* + \int_{-\infty}^{\infty}\left|z\text{Re}\left(\left(u'+i\omega_0u\right)\overline{H}\right)\right|dr^*.
\end{equation}
As usual, this implies
\[\int_{r_+}^{r_++\epsilon}(r-r_+)^{-2}|u'(r^*(r))+i\omega_0u(r^*(r))|^2dr + \int_{r_+}^{r_++\epsilon/2}\left|u(r^*(r))\right|^2dr \lesssim \]
\begin{equation}\label{strictredShiftEstimate}
\int_{\epsilon/2}^{2\epsilon}\left(\left|u(r^*(r))\right|^2 + \left|u'(r^*(r))\right|^2\right)dr^* + \int_{-\infty}^{\infty}\left|F(r)\right|^2dr.
\end{equation}
Note that for every fixed $m$ and $l$, $\epsilon$ can be assumed to depend continuously on $a$ and $\omega$. This estimate is good near $-\infty$, but clearly is not sufficient otherwise.
\subsection{Proof of Proposition \ref{secondBoundEstimate}}
Let $b_0 \in (r_+,\infty)$ be sufficiently close to $r_+$ and $b_1 = 2b_0$. First we apply virial estimate I and conclude
\begin{equation}\label{something}
\int_{b_0}^{\infty}\left(\left|R\right|^2 + \left|\frac{dR}{dr}\right|^2\right)dr \lesssim \left|Q_T(\infty)\right| + \int_{r_+}^{\infty}|F|^2r^4dr.
\end{equation}
Now, depending on whether $\omega_0$ is small or large, we either carry out virial estimate II or the red-shift estimate and combine with \ref{something} to get
\[\int_{r_+}^{\infty}\left|R\right|^2dr \lesssim \left|Q_T(-\infty)\right| + \left|Q_T(\infty)\right| + \int_{r_+}^{\infty}\left|F\right|^2dr.\]
Next, we recall that the energy current $Q_T = \omega\text{Im}\left(u'\overline{u}\right)$ satisfies
\[\left(Q_T\right)' = \omega\text{Im}\left(H\overline{u}\right) \Rightarrow\]
\[\left|Q_T(\infty)\right| \leq \left|Q_T(-\infty)\right| + \int_{r_+}^{\infty}(r^2+a^2)\left|F\right|\left|R\right|dr \lesssim \]
\[\epsilon^{-1}\int_{r_+}^{\infty}\left|F(r)\right|^2r^4dr + \epsilon\int_{-\infty}^{\infty}\left|R(r)\right|^2dr\]
where we have used (\ref{theEst}) in the last line. Taking $\epsilon$ small enough, we may combine the various estimates to conclude
\[\int_{r_+}^{\infty}\left|R\right|^2dr \lesssim \int_{r_+}^{\infty}\left|F(r)\right|^2r^4dr.\]
Reapplying the energy estimate finally implies
\[\left|u(-\infty)\right|^2 \approx \left|Q_T(-\infty)\right| \lesssim \int_{r_+}^{\infty}\left|F(r)\right|^2r^4dr.\]
\section{A Unique Continuation Lemma}\label{vanish1}
\begin{lemm}\label{vanish}Suppose that we have a solution $u(r^*) : (-\infty,\infty) \to \mathbb{C}$ to the ODE
\[u'' + \left(\omega^2 - V\right)u = 0\]
such that
\begin{enumerate}
    \item $\omega \in \mathbb{R}\setminus \{0\}$,
    \item $u \in L^{\infty}$ and $\left(\left|u'\right|^2 + \left|u\right|^2\right)\left(\infty\right) = 0$,
    \item $V$ is real, $V \in L^{\infty}$, $V = O\left(r^{-1}\right)\text{ as }r\to\infty$, and $V' = O\left(r^{-2}\right)\text{ as }r\to\infty$.
\end{enumerate}
Then $u$ is identically $0$.
\end{lemm}
\begin{proof}We will slightly refine the estimate from section \ref{virialEstimate}.\footnote{Using that estimate directly would require that $V = O\left(r^{-2}\right)$ as $r\to\infty$.} Define
\[y(r^*) := \exp\left(-B\int_{r^*}^{\infty}\zeta(r)dr\right)\]
where $B$ is a large positive constant to be chosen later and $\zeta$ is a \emph{fixed} positive function which is identically $1$ near $r = -\infty$ and is equal to $r^{-2}$ near $r = \infty$. In particular, we have $y'|_{\left(-\infty,\infty\right)} > 0$, $y(-\infty) = 0$, and $y(\infty) = 1$.

Next, set
\[Q^y(r^*) := y\left|u'\right|^2 + y\left(\omega^2-V\right)\left|u\right|^2.\]
Observe that the hypothesis of the lemma imply that $Q^y\left(\pm\infty\right) = 0$. A simple computation gives
\[\left(Q^y\right)' = y'\left|u'\right|^2 + y'\omega^2\left|u\right|^2 - \left(yV\right)'\left|u\right|^2.\]
Thus, the fundamental theorem of calculus implies
\begin{equation}\label{theInteriorTerm}
\int_{-\infty}^{\infty}\left(y'\left|u'\right|^2 + y'\omega^2\left|u\right|^2 - \left(yV\right)'\left|u\right|^2\right)dr^* = 0.
\end{equation}
Let $R \in (1,\infty)$ be a large constant to be chosen later. Then set $\chi(r^*)$ to be a function identically $1$ on $(-\infty,R]$ and $0$ on $[R+1,\infty)$. We then define
\[V_1 := \chi V,\]
\[V_2 := (1-\chi)V.\]
Of course we have $V = V_1 + V_2$.

We have the following estimate:
\begin{align}
\left|\int_{-\infty}^{\infty}\left(yV_1\right)'\left|u\right|^2dr^*\right| &= 2\left|\int_{-\infty}^{\infty}yV_1\text{Re}\left(u'\overline{u}\right)dr^*\right|\\
&\leq \epsilon \int_{-\infty}^{\infty}y'\left|u'\right|^2dr^* + \epsilon^{-1}\int_{-\infty}^{\infty}y'\omega^2\left(\frac{y^2V_1^2}{\left(y'\right)^2\omega^2}\right)\left|u\right|^2dr^* \\
&\leq\label{firstIntegration} \epsilon \int_{-\infty}^{\infty}y'\left|u'\right|^2dr^* + C\epsilon^{-1}\omega^{-2}B^{-2}R^2\int_{-\infty}^{\infty}y'\omega^2\left|u\right|^2dr^*.
\end{align}
Here $C$ is a constant which only depends on $\zeta$ and $V$.

Next we estimate
\begin{align}
\left|\int_{-\infty}^{\infty}\left(yV_2\right)'\left|u\right|^2dr^*\right| &= \left|\int_{-\infty}^{\infty}\left(y'V_2 + yV_2'\right)\left|u\right|^2dr^*\right|\\
&\leq\label{secondIntegration} C\int_{-\infty}^{\infty}\left(R^{-1}\omega^{-2} + B^{-1}\omega^{-2}\right)y'\omega^2\left|u\right|^2dr^*.
\end{align}
Taking $\epsilon$ small, $R$ large, and then $B$ sufficiently large and combining (\ref{theInteriorTerm}), (\ref{firstIntegration}), and (\ref{secondIntegration}) implies that
\[\frac{1}{2}\int_{-\infty}^{\infty}\left(y'\left|u'\right|^2 + y'\omega^2\left|u\right|^2\right)dr^* = 0.\]
\end{proof}
\section{Acknowledgments}
I thank my advisor Igor Rodnianski for suggesting the problem and for many insightful conversations. I would also like to thank Mihalis Dafermos for very useful comments on preliminary versions of the paper.
\appendix
\section{Asymptotic Analysis of the Radial ODE}\label{odeFacts}
We will collect various facts concerning the radial ODE:
\[u'' + \left(\omega^2-V\right)u = 0\text{ for }\omega \in \mathbb{C}\setminus\{0\}.\]
The material in this section is standard, and the necessary background can be found in most textbooks on the asymptotic analysis of ODE's, e.g. \cite{n7}.

When recast in the $r$ variable our ODE has a regular singularity at $r_+$. Finding the roots of the indicial equation allows us to uniquely define two linearly independent functions $u_{\text{hor}}$ (already given by definition \ref{theHor}) and $u_{\text{hor2}}$ by
\begin{defi}Let $u_{\text{hor2}}(r^*)$ be the unique function satisfying
    \begin{enumerate}
        \item $u_{\text{hor2}}'' + \left(\omega^2-V\right)u_{\text{hor2}} = 0$.
        \item\label{asympt} $u_{\text{hor2}} \sim (r-r_+)^{\frac{-i(am-2Mr_+\omega)}{r_+-r_-}}\text{ near }r^* = -\infty$.
        \item $\left|(r-r_+)^{\frac{i(am-2Mr_+\omega)}{r_+-r_-}}u_{\text{hor2}}\left(-\infty\right)\right|^2 = 1$.
    \end{enumerate}
\end{defi}
Since we have a regular singularity, the ``$\sim$'' means that \[u_{\text{hor}}(r^*)(r-r_+)^{\frac{-i(am-2Mr_+\omega)}{r_+-r_-}}\]
is holomorphic in $r$ near $r_+$. In fact, it can be given by an explicit power series which exhibits holomorphic dependence on $\omega$. Analogous statements hold for $u_{\text{hor2}}$.

Our ODE has an irregular singularity at $\infty$. Nevertheless, we can uniquely define two linearly independent functions $u_{\text{in}}$ and $u_{\text{out}}$ (already given by definition \ref{theOut}) by
\begin{defi}Let $u_{\text{in}}(r^*)$ be the unique function satisfying
    \begin{enumerate}
        \item $u_{\text{in}}'' + \left(\omega^2-V\right)u_{\text{out}} = 0$.
        \item $u_{\text{in}} \sim e^{-i\omega r^*}\text{ near }r^* = \infty$.
        \item $\left|\left(e^{i\omega r^*}u_{\text{in}}\right)\left(\infty\right)\right|^2 = 1$.
    \end{enumerate}
\end{defi}
Since our singularity is irregular, ``$\sim$'' must be interpreted as follows: There exists explicit constants $\left\{C^{(\text{in})}_i\right\}_{i=1}^{\infty}$ and $\left\{C^{(\text{out})}_i\right\}_{i=1}^{\infty}$ such that for every $N \geq 1$
\[u_{\text{in}}(r^*) = e^{-i\omega r^*}\left(1 + \sum_{i=1}^N\frac{C^{\text{in}}_i}{(r^*)^i}\right) + O\left(\left(r^*\right)^{-N-1}\right)\text{ for large }r^*,\]
\[u_{\text{out}}(r^*) = e^{i\omega r^*}\left(1 + \sum_{i=1}^N\frac{C^{\text{out}}_i}{(r^*)^i}\right) + O\left(\left(r^*\right)^{-N-1}\right)\text{ for large }r^*.\]
It is important to note that the constants in these $O$'s can be estimated explicitly if desired. By examining the construction of $u_{\text{out}}$, one finds that $u_{\text{out}}$ will be holomorphic in $\omega$ in the upper half plane and smooth in $\omega$ in $\mathbb{R}\setminus\{0\}$. See \cite{n17} for a detailed discussion of the holomorphic dependence on $\omega$.

\section{Energy Currents}\label{theFirstt}
It will be useful to use the language of energy currents which we now briefly review (see \cite{n8} for a proper introduction).

Fix a smooth function $\psi$. The energy-momentum tensor is given by
\[T_{\alpha\beta} := \text{Re}\left(\partial_{\alpha}\psi\overline{\partial_{\beta}\psi}\right) - \frac{1}{2}g_{\alpha\beta}g^{\gamma\delta}\text{Re}\left(\partial_{\gamma}\psi\overline{\partial_{\delta}\psi}\right).\]
Given any vector field $X$ we form the corresponding ``current'' by
\[J^X_{\alpha} := T_{\alpha\beta}X^{\beta}.\]
We have
\begin{lemm}\label{energyMomBreakDown}Let $X$ and $Y$ be two linearly independent future oriented timelike vectors normalized to have $g(X,X) = g(Y,Y) = -1$. Set $\gamma := -g(X,Y)$. Note that $\gamma > 1$ by the reverse Cauchy-Schwarz inequality. Define
\[W := \frac{1}{\sqrt{2\left(\gamma+1\right)}}\left(X+Y\right),\]
\[Z := \frac{1}{\sqrt{2\left(\gamma-1\right)}}\left(X-Y\right),\]
\[L := W + Z,\]
\[\underline{L} := W- Z.\]
Let $E_1$ and $E_2$ be an orthonormal basis in the $2$-dimensional subspace orthogonal to the span of $X$ and $Y$. Then,
\[J^X_{\alpha}Y^{\alpha} = \frac{1}{4}\left(\left|L\psi\right|^2 + \left|\underline{L}\psi\right|^2\right) + \frac{\gamma}{2}\left(\left|E_1\psi\right|^2 + \left|E_2\psi\right|^2\right).\]
\end{lemm}
\begin{proof}Observe that
\[g\left(L,L\right) = g\left(\underline{L},\underline{L}\right) = 0,\]
\[g\left(L,\underline{L}\right) = -2,\]
\[X = (1/4)\left(\left(\sqrt{2\left(\gamma+1\right)} + \sqrt{2\left(\gamma-1\right)}\right)L + \left(\sqrt{2\left(\gamma+1\right)} - \sqrt{2\left(\gamma-1\right)}\right)\underline{L}\right),\]
\[Y = (1/4)\left(\left(\sqrt{2\left(\gamma+1\right)} - \sqrt{2\left(\gamma-1\right)}\right)L + \left(\sqrt{2\left(\gamma+1\right)} + \sqrt{2\left(\gamma-1\right)}\right)\underline{L}\right).\]
The result then follows from a simple computation using the algebraic properties of the energy-momentum tensor (see \cite{n8}).
\end{proof}
It is also possible to find a convenient expression for $J^X_{\alpha}X^{\alpha}$.
\begin{lemm}Let $X$ be a timelike vector normalized to have $g\left(X,X\right) = -1$. Let $R$ be any spacelike vector orthogonal to $X$, normalized to have size $1$. Define
\[L := X + R,\]
\[\underline{L} := X-R.\]
Let $E_1$ and $E_2$ be an orthonormal basis for the subspace orthogonal to the span of $X$ and $R$. Then
\[J^X_{\alpha}X^{\alpha} = \left|L\psi\right|^2 + \left|\underline{L}\psi\right|^2 + \left|E_1\psi\right|^2 + \left|E_2\psi\right|^2.\]
\end{lemm}
\begin{proof}This is a simple computation using the algebraic properties of the energy-momentum tensor (see \cite{n8}).
\end{proof}
This leads to
\begin{defi}Let $X$ be a future oriented timelike vector field and $\Sigma$ be a spacelike hypersurface with future oriented normal $n_{\Sigma}$. We define the (non-degenerate) energy of $\psi$ with respect to $X$ along $\Sigma$ by
\begin{equation}\label{theEnergy}
\int_{\Sigma}J^X_{\alpha}n_{\Sigma}^{\alpha}
\end{equation}
where the integral is with respect to the induced volume form. We will often use the schematic notation
\[\int_{\Sigma}\left|\partial\psi\right|^2\]
to denote (\ref{theEnergy}).
\end{defi}
\section{From Admissibility to Uniform Boundedness}\label{admissibleBounded}
The following lemma is a straightforward application of techniques developed in \cite{n3} (see also the lecture notes \cite{n6} and sections 5.4 and 5.5 of \cite{n4}).
\begin{lemm}\label{easyBound}Suppose that $\psi$ is admissible in the sense of definition \ref{admissible} and that $\psi$ solves the wave equation $\Box_g\psi = 0$ on the Kerr spacetime. Then $\psi$ is uniformly bounded to the future.
\end{lemm}
\begin{proof}Fix a choice of $r_0$ sufficiently close to but greater than $r_+$, and $r_1 < \infty$ sufficiently large. Our admissibility assumption implies
\[\int_0^{\infty}\int_{r\in (r_0,r_1)}\int_{\mathbb{S}^2}\left|\partial\psi\right|^2\sin\theta\, dt\, dr\, d\theta\, d\phi < \infty.\]
Let $V$ be a timelike and time translation invariant vector field in the region $r > r_0$ which equals $\partial_t$ for $r > r_1/2$. Let $\delta > 0$ be sufficiently small. By cutting off $\psi$ to the region $r > r_0$ and applying the energy estimate associated to $V$ (see \cite{n8}), we conclude that for every $\tau > 0$,
\[\int_{r > r_0+\delta}\int_{\mathbb{S}^2}\left|\partial\psi\right|^2\Big|_{t = \tau} r^2\sin\theta\, dr\, d\theta\, d\phi \lesssim_{r_0,r_1,\delta}\]
\[ \int_0^{\infty}\int_{r \in (r_0,r_1/2)}\int_{\mathbb{S}^2}\left|\partial\psi\right|^2\sin\theta\, dt\, dr\, d\theta\, d\phi + \]
\[\int_{r > r_0}\int_{\mathbb{S}^2}\left|\partial\psi\right|^2\Big|_{t = 0} r^2\sin\theta\, dr\, d\theta\, d\phi.\footnote{The point being that $V$ is Killing for $r > r_1/2$, and hence no spacetime error terms arise in that region.}\]
In order to control $\psi$ in the region $r < r_0$ we shall appeal to the \emph{red-shift} estimate of Dafermos-Rodnianski (see \cite{n3}, \cite{n6}, and \cite{n19}). As long as $r_0$ is sufficiently close to $r_+$, this estimate implies
\[\int_{\Sigma_{\tau}\cap\{r \leq r_0 + \delta\}}\left|\partial\psi\right|^2 + \int_0^{\tau}\int_{\Sigma_{t}\cap\{r \leq r_0 + \delta\}}\left|\partial\psi\right|^2\ dt \lesssim \]
\[\int_0^{\infty}\int_{r \in (r_0,r_1)}\int_{\mathbb{S}^2}\left|\partial\psi\right|^2\sin\theta\, dr\, d\theta\, d\phi + \]
\[\int_{\Sigma_0\cap\{r \leq r_1\}}\left|\partial\psi\right|^2.\footnote{The red-shift estimate also gives a good term on the horizon, but we do not need this.}\]
Let us emphasize that $\left|\partial\psi\right|^2$ denotes a \emph{non-degenerate} energy flux (see the discussion in appendix \ref{theFirstt}) and the integration is with respect to the induced volume form. Here $\Sigma_{\tau}$ refers to the time translations of the hypersurface $\Sigma_0$ from Theorem \ref{estHorizon}.

Given such a uniform bound on the non-degenerate energy, bounds on higher order energies follows in a standard fashion by commuting with the Killing vector fields $\partial_t$ and (a cut-off version of) $\partial_{\phi}$, commuting with the red-shift commutator $\hat{Y}$, and finally applying elliptic estimates. Once the higher order energies are controlled, pointwise boundedness follows from Sobolev inequalities. In section 13 of \cite{n3} one can find this scheme carried out in full detail for the case of $\left|a\right| \ll M$. A direct inspection of the argument there shows that the only difference in the case $\left|a\right| < M$ is that one also needs to commute with (a cut-off version of) $\partial_{\phi}$.\footnote{The key point being that in the domain of outer communication, there is always a timelike direction in the span of $\partial_t$ and $\partial_{\phi}$.} This fact is explicitly discussed in section 5.5 of \cite{n4}.
\end{proof}
\section{Modes and Their Finite Energy Hypersurfaces}\label{finiteModeEnergy}
In this appendix we will explore the hypersurfaces on which various modes have finite energy.
\subsection{The Hypersurfaces}
For purposes of exposition we will restrict attention to spacelike hypersurfaces $\Sigma_f$ which, for sufficiently large $R$, satisfy
\[\Sigma_f \cap \{r \geq R\} := \left\{\left(t,r^*,\theta,\phi\right) : r\geq R\text{ and } t - f(r^*) = 0\right\}.\]
In addition to the requirement that $\Sigma_f$ be spacelike, we also ask that $\Sigma_f$ intersects the future event horizon and
\[f \geq 0\text{ as }r^* \to \infty. \]
This last requirement implies that $\Sigma_f$ connects the event horizon $\mathcal{H}^+$ to either spacelike infinity or future null infinity.
\begin{defi}We will say that $\Sigma_f$ is asymptotically flat if $f \sim 1$ as $r^* \to \infty$.\footnote{More generally, one could consider any hypersurface which terminates at spacelike infinity, but this extra generality is not useful for the study of mode solutions.}
\end{defi}
These hypersurfaces converge to spacelike infinity as $r^*\to\infty$. The prototypical example of an asymptotically flat hypersurface is one where $f$ is identically constant for large $r$. The relevant Penrose diagram is
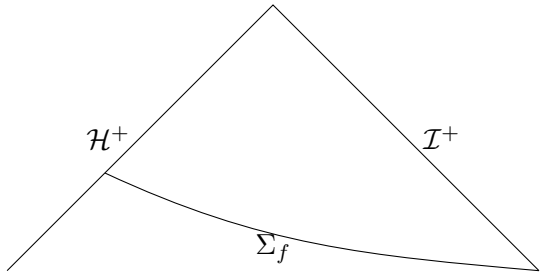
\begin{center}
    \begin{tikzpicture}
        [interior/.style={circle,draw=black,fill=black, inner sep=0pt,minimum size = 2.5mm},
        boundary/.style={circle,draw=black,fill=black!60, inner sep=0pt,minimum size = 2.5mm},
        exterior/.style={circle,draw=black,fill=white, inner sep=0pt,minimum size = 2.5mm},
        highlight/.style = {circle,draw= red, fill = red, inner sep=0pt, minimum size = 2.5mm}]
        \draw (0,0) -- (3.535,3.535) -- (7.071,0);
        \node [left] at (1.768,1.768) {$\mathcal{H}^+$};
        \node [right] at (5.4,1.768) {$\mathcal{I}^+$};
        \draw (1.3,1.3) to [out = -25,in = 175] (7.071,0);
        \node [above] at (3.536,-.03) {$\Sigma_f$};
    \end{tikzpicture}
\end{center}

\begin{defi}We will say that $\Sigma_f$ is hyperboloidal if $\left(f'\right)^2 - 1 = -\frac{C}{r^2} + O\left(r^{-3}\right)$ as $r^* \to\infty$ for some sufficiently large positive constant $C$ ($C \geq M$ will work).\footnote{In more general contexts one usually says a spacelike hypersurface is hyperboloidal if the induced metric asymptotically approaches a constant negative curvature metric. One could work with this more general definition here; but, since there is not much advantage \`a la the study of mode solutions, we shall spare ourselves the extra work.}
\end{defi}
These hypersurfaces converge to future null infinity as $r^*\to\infty$. The key examples to keep in mind are hyperbolas in Minkowski space (where $f = \sqrt{C + r^2}$). The relevant Penrose diagram is
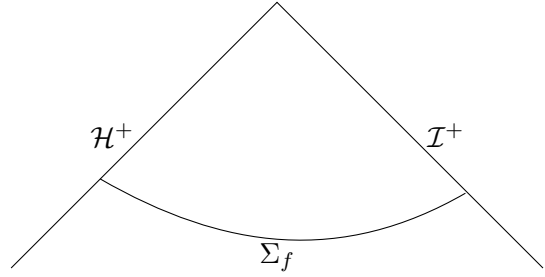
\begin{center}
    \begin{tikzpicture}
        [interior/.style={circle,draw=black,fill=black, inner sep=0pt,minimum size = 2.5mm},
        boundary/.style={circle,draw=black,fill=black!60, inner sep=0pt,minimum size = 2.5mm},
        exterior/.style={circle,draw=black,fill=white, inner sep=0pt,minimum size = 2.5mm},
        highlight/.style = {circle,draw= red, fill = red, inner sep=0pt, minimum size = 2.5mm}]
        \draw (0,0) -- (3.535,3.535) -- (7.071,0);
        \node [left] at (1.768,1.768) {$\mathcal{H}^+$};
        \node [right] at (5.4,1.768) {$\mathcal{I}^+$};
        \draw (1.184,1.186) to [out = -30,in = 210] (6.04,.996);
        \node [below] at (3.536,.47) {$\Sigma_f$};
    \end{tikzpicture}
\end{center}

\subsection{Some Useful Calculations}\label{theSecond}
We start by noting that
\[g^{tt} = \frac{a^2\sin^2\theta\Delta - (r^2+a^2)^2}{\rho^2\Delta},\]
\[g^{\phi\phi} = \frac{\Delta - a^2\sin^2\theta}{\sin^2\theta\Delta},\]
\[g^{t\phi} = -\frac{4Mar}{\rho^2\Delta}.\]

Then we have
\begin{lemm}\label{theFirst}Let $\Sigma_f$ be an asymptotically flat hypersurface, $N$ be a future oriented timelike vector field which equals $\partial_t$ for large $r$, and $\psi$ be a smooth function. Then, for sufficiently large $R$, the energy of $\psi$ with respect to $N$ along $\Sigma_f\cap\{r \geq R\}$ is proportional to
\[\int_{r \geq R}\int_{\mathbb{S}^2}\left(\left|\partial_t\psi\right|^2 + \left|\partial_r\psi\right|^2 + r^{-2}\left(\left(\partial_{\theta}\psi\right)^2 + \sin^{-2}\theta\left(\partial_{\phi}\psi\right)^2\right)\right)\left(f(r^*),r,\theta,\phi\right)\times \]
\[r^2\sin\theta\ dr\ d\theta\ d\phi.\]
\end{lemm}
\begin{proof}First, observe that
\[-\nabla t = -g^{tt}\partial_t - g^{t\phi}\partial_{\phi} = \frac{(r^2+a^2)^2 - a^2\sin^2\theta}{\rho^2\Delta}\partial_t + \frac{4Mar}{\rho^2\Delta}\partial_{\phi},\]
\[g\left(\nabla t,\nabla t\right) = \frac{-(r^2+a^2)^2 + a^2\sin^2\theta\Delta}{\rho^2\Delta}.\]
In particular, $\nabla t$ is timelike. Next, we calculate
\[g\left(\nabla\left(t-f(r^*)\right),\nabla\left(t-f(r^*)\right)\right) = \]
\[\left(\left(f'\right)^2-1\right)\frac{(r^2+a^2)^2}{\rho^2\Delta} + \frac{a^2\sin^2\theta}{\rho^2} \to -1\text{ as }r\to\infty.\]
We conclude that the normal to $\Sigma_f$ satisfies
\[n_{\Sigma_f} = \left(1 + O\left(r^{-1}\right)\right)\left(-\nabla t\right) + O\left(r^{-1}\right)\partial_{r^*}\text{ as }r\to\infty.\]
Now, lemma \ref{energyMomBreakDown} implies
\[J^N_{\alpha}n^{\alpha}_{\Sigma_f} \approx \left|\partial_t\psi\right|^2 + \left|\partial_{r^*}\psi\right|^2 + r^{-2}\left(\left(\partial_{\theta}\psi\right)^2 + \sin^{-2}\theta\left(\partial_{\phi}\psi\right)^2\right)\text{ as }r\to\infty.\]
The volume form on Kerr satisfies
\[dVol = \frac{\Delta\rho^2}{r^2+a^2}\sin\theta\ dt\wedge dr^*\wedge d\theta\wedge d\phi.\]
Thus, the induced volume on $\Sigma_f$ is given by
\[\left(1 + O(r^{-1})\right)r^2\sin\theta\ dr^*\wedge d\theta\wedge d\phi + \left(1+O(r^{-1})\right)r\sin\theta\ dt\wedge\ d\theta\wedge d\phi + \] \[\left(1+O(r^{-1})\right)r\sin\theta\ dt\wedge dr^*\wedge d\theta\text{ as }r\to\infty.\]
The lemma follows by writing out the integral (\ref{theEnergy}) in the parametrization $\left(r^*,\theta,\phi\right) \mapsto \left(f(r^*),r^*,\theta,\phi\right)$.
\end{proof}

The analogous lemma in the hyperboloidal case is more subtle since we need to understand precisely how the energy degenerates due to the hypersurface becoming ``approximately null.''
\begin{lemm}\label{theSecond}Let $\Sigma_f$ be a hyperboloidal hypersurface, $N$ be a future oriented timelike vector field which equals $\partial_t$ for large $r$, and $\psi$ be a smooth function. Then, for sufficiently large $R$, the energy of $\psi$ with respect to $N$ along $\Sigma_f\cap\{r \geq R\}$ is proportional to
\begin{align*}
\int_{r\geq R}\int_{\mathbb{S}^2}&\Bigg(r^{-2}\left|\left(\partial_t- \partial_{r^*}\right)\psi\right|^2 + \left|\left(\partial_t+\partial_{r^*}\right)\psi\right|^2 + \\ &r^{-2}\left(\left(\partial_{\theta}\psi\right)^2 + \sin^{-2}\theta\left(\partial_{\phi}\psi\right)^2\right)\Bigg)r^2\sin\theta\ dr\ d\theta\ d\phi
\end{align*}
where the integrand is evaluated at $\left(f(r^*),r^*,\theta,\phi\right)$.
\end{lemm}
\begin{proof}Let's set
\begin{align*}
A_{\Sigma_f} &:= \sqrt{-g\left(\nabla\left(t-f(r^*)\right),\nabla\left(t-f(r^*)\right)\right)}\\ &=\sqrt{\left(1-\left(f'\right)^2\right)\frac{(r^2+a^2)^2}{\rho^2\Delta} - \frac{a^2\sin^2\theta}{\rho^2}}\\
&= O\left(r^{-1}\right)\text{ as }r\to\infty.
\end{align*}
The normal $n_{\Sigma_f}$ thus satisfies
\[n_{\Sigma_f} = A_{\Sigma_f}^{-1}\left(-\nabla t + f'\frac{(r^2+a^2)^2}{\rho^2\Delta}\partial_{r^*}\right).\]
The key difference with the asymptotically flat case is that $A_{\Sigma_f}^{-1} = r + O(1)$ as $r\to\infty$.

Let's apply lemma \ref{energyMomBreakDown} to the vectors $X := \left(-g_{tt}\right)^{-1/2}\partial_t$ and $Y := n_{\Sigma_f}$. We have
\[\gamma := -g\left(X,Y\right) = \left(-g_{tt}\right)^{-1/2}A_{\Sigma_f}^{-1} = r + O(1)\text{ as }r\to\infty.\]
Next, we compute
\begin{align*}
W &= \frac{1}{\sqrt{2\left(\gamma+1\right)}}\left(X+Y\right)\\
  &= O\left(r^{-1/2}\right)\partial_t + \left(r^{1/2}+O\left(r^{-1/2}\right)\right)\left(-\nabla t + \partial_{r^*}\right)\text{ as }r\to\infty,\\ \\
Z &= \frac{1}{\sqrt{2\left(\gamma-1\right)}}\left(X-Y\right)\\
  &= O\left(r^{-1/2}\right)\partial_t - \left(r^{1/2}+O\left(r^{-1/2}\right)\right)\left(-\nabla t + \partial_{r^*}\right)\text{ as }r\to\infty,\\ \\
L &= W + Z \\
  &= O\left(r^{-1/2}\right)\left(\partial_t + \left(-\nabla t + \partial_{r^*}\right)\right)\text{ as }r\to\infty,\\ \\
\underline{L} &= W - Z \\
              &= O\left(r^{-3/2}\right)\partial_t + 2\left(r^{1/2}+O\left(r^{-1/2}\right)\right)\left(-\nabla t + \partial_{r^*}\right)\text{ as }r\to\infty.
\end{align*}
Finally, as $r\to\infty$, the induced volume form satisfies
\[\left(A_{\Sigma_f}^{-1}\frac{(r^2+a^2)^2}{\rho^2\Delta} + O\left(r^{-3}\right)\right)\left(\frac{\Delta\rho^2}{r^2+a^2}\right)\sin\theta\ dr^*\wedge d\theta\wedge d\phi  \]
\[+O\left(1\right)\frac{\Delta\rho^2}{r^2+a^2}\sin\theta\ dt\wedge dr^*\wedge d\theta \]
\[-\left(A_{\Sigma_f}^{-1}f'\frac{(r^2+a^2)^2}{\rho^2\Delta}\right)\left(\frac{\Delta\rho^2}{r^2+a^2}\right)\sin\theta\ dt\wedge d\theta\wedge d\phi.\]
The lemma now follows by carefully writing out the integral (\ref{theEnergy}) in the parametrization $\left(r^*,\theta,\phi\right)\mapsto \left(f(r^*),r^*,\theta,\phi\right)$, using $\left(f'\right)^2 - 1 = -\frac{C}{r^2} + O\left(r^{-3}\right)$, and appealing to lemma \ref{energyMomBreakDown}.
\end{proof}
\subsection{Finite Energy Hypersurfaces for Mode Solutions}
\begin{lemm}\label{modeThing}Let $\Sigma_f$ be an asymptotically flat hypersurface, $N$ be a future oriented timelike vector field which equals $\partial_t$ for large $r$, and
\[\psi(t,r,\theta,\phi) = e^{-i\omega t}e^{im\phi}S_{\omega ml}(\theta)R(r)\]
be a mode solution. If $\text{Im}\left(\omega\right) > 0$ then $\psi$ has finite energy with respect to $N$ along $\Sigma_f$. If $\text{Im}\left(\omega\right) \leq 0$ then $\psi$ has infinite energy with respect to $N$ along $\Sigma_f$.
\end{lemm}
\begin{proof}In Kerr-star coordinates, it is easy to see that the volume form remains bounded in a compact region of $r$ (including the event horizon). Thus, in order for $\psi$ to have finite energy along $\Sigma_f \cap \{r \leq R\}$ it is sufficient for $\psi$ to be smooth (and hence bounded). Furthermore, $\psi$ is manifestly smooth if $r > r_+$. Since Boyer-Lindquist coordinates break down at $r = r_+$, in order to investigate the smoothness of $\psi$ there, we will change to Kerr-star coordinates $(t^*,r,\theta,\phi^*)$. In these coordinates we get
\[\psi(t^*,r,\theta,\phi^*)  = e^{-i\omega\left(t^*-\overline{t}(r)\right)}e^{im\left(\phi^*-\overline{\phi}(r)\right)}S_{\omega ml}(\theta)R(r).\]
Hence, $\psi$ extends smoothly to $r=r_+$ if and only if
\[R(r) = e^{-i\left(\omega \overline{t}(r) - m\overline{\phi}(r)\right)}h(r)\]
where $h$ extends smoothly to $r_+$. However, this is precisely what the boundary condition (\ref{horizonBound}) guarantees.

For $R$ sufficiently large, lemma \ref{theFirst} imply that the energy along $\Sigma_f \cap \{r \geq R\}$ is proportional to
\[\int_{r \geq R}\int_{\mathbb{S}^2}\left(\left|\partial_t\psi\right|^2 + \left|\partial_r\psi\right|^2 + r^{-2}\left(\left(\partial_{\theta}\psi\right)^2 + \sin^{-2}\theta\left(\partial_{\phi}\psi\right)^2\right)\right)\left(f(r^*),r,\theta,\phi\right)\times \]
\[r^2\sin\theta\ dr\ d\theta\ d\phi.\]
Now, if $\text{Im}\left(\omega\right) > 0$, then the boundary condition (\ref{infinityBound}) implies that all of these terms are decaying exponentially as $r \to \infty$, and hence, the integral is finite. If $\text{Im}\left(\omega\right) = 0$, then the first two terms in the integral as proportional to $r^{-2}$, and hence the integral is infinite. If $\text{Im}\left(\omega\right) < 0$, then all of the terms are exponentially growing in $r$, and hence the integral is infinite.
\end{proof}
\begin{lemm}Let $\Sigma_f$ be a hyperboloidal hypersurface, $N$ be a future oriented timelike vector field which equals $\partial_t$ for large $r$, and
\[\psi(t,r,\theta,\phi) = e^{-i\omega t}e^{im\phi}S_{\omega ml}(\theta)R(r)\]
be a mode solution with $\text{Im}\left(\omega\right) \leq 0$. Then $\psi$ has finite energy with respect to $N$ along $\Sigma_f$.
\end{lemm}
\begin{proof} The analysis of $\psi$ for any compact region of $r$ is exactly the same as in the proof of lemma \ref{modeThing}. In lemma \ref{theSecond} we saw that the energy along $\Sigma_f \cap \{r \geq R\}$ is proportional to
\begin{align*}
\int_{r\geq R}\int_{\mathbb{S}^2}&\Bigg(r^{-2}\left|\left(\partial_t- \partial_{r^*}\right)\psi\right|^2 + \left|\left(\partial_t+\partial_{r^*}\right)\psi\right|^2 + \\ &r^{-2}\left(\left(\partial_{\theta}\psi\right)^2 + \sin^{-2}\theta\left(\partial_{\phi}\psi\right)^2\right)\Bigg)r^2\sin\theta\ dr\ d\theta\ d\phi
\end{align*}
where the integrand is evaluated at $\left(f(r^*),r^*,\theta,\phi\right)$.

When $\text{Im}\left(\omega\right) = 0$, then the boundary condition \ref{infinityBound} exactly implies that $\left(\partial_t+\partial_{r^*}\right)\psi = O\left(r^{-2}\right)$. Combining this with the fact that $\psi$ and it's derivatives are all $O\left(r^{-1}\right)$ shows that the integral is finite.

Now consider the case where $\text{Im}\left(\omega\right) < 0$. Using the boundary condition \ref{infinityBound}, we get
\begin{align*}
\psi\left(f(r^*),r^*,\theta,\phi\right) &= \exp\left(-i\omega f\left(r^*\right)\right)e^{im\phi}S_{\omega ml}\left(\theta\right)R\left(r\right)\\
                                        &= O\left(r^{-1}\exp\left(-i\omega r^*\right)\exp\left(i\omega\left(r^* - 2M\log r\right)\right)\right)\text{ as }r\to\infty\\
                                        &= O\left(r^{-1}\right)\text{ as }r\to\infty.
\end{align*}
Similarly,
\[\partial_t\psi\left(f(r^*),r^*,\theta,\phi\right) = O\left(r^{-1}\right)\text{ as }r\to\infty,\]
\[\partial_{r^*}\psi\left(f(r^*),r^*,\theta,\phi\right) = O\left(r^{-1}\right)\text{ as }r\to\infty,\]
\[\partial_{\theta}\psi\left(f(r^*),r^*,\theta,\phi\right) = O\left(r^{-1}\right)\text{ as }r\to\infty,\]
\[\partial_{\phi}\psi\left(f(r^*),r^*,\theta,\phi\right) = O\left(r^{-1}\right)\text{ as }r\to\infty,\]
\[\left(\partial_t+\partial_{r^*}\right)\psi\left(f(r^*),r^*,\theta,\phi\right) = O\left(r^{-2}\right)\text{ as }r\to\infty.\]
Thus, the integral is finite.

\end{proof}

\end{document}